\documentclass[copyright,creativecommons]{eptcs}
\providecommand{\event}{QPL 2025} 
\pdfoutput=1
\usepackage[normalem]{ulem}
\usepackage{comment}

\usepackage{amsthm}
\usepackage{mathtools}
\usepackage{bbm}
\usepackage{bm}
\usepackage{dsfont}
\usepackage{color}
\usepackage{graphicx}
\usepackage{amssymb}
\usepackage{float}
\usepackage{physics}
\usepackage{tikz}

\usepackage{iftex}

\ifpdf
  \usepackage{underscore}         
  \usepackage[T1]{fontenc}        
\else
  \usepackage{breakurl}           
\fi

\usepackage{thmtools}
\usepackage{thm-restate}

\newcommand{\titlerunning}{Classical and Quantum Query Complexity of Boolean Functions under Indefinite Causal Order}
\newcommand{\authorrunning}{A.\ A.\ Abbott, M.\ Mhalla \& P.\ Pocreau}

\newtheorem{theorem}{Theorem}
\newtheorem{lemma}[theorem]{Lemma}
\newtheorem{corollary}[theorem]{Corollary}
\newtheorem{proposition}[theorem]{Proposition}
\theoremstyle{definition}
\newtheorem{definition}[theorem]{Definition}

\newenvironment{sproof}{%
  \proof}{\endproof}

\newcommand{\id}{\mathds{1}}
\renewcommand{\H}{\mathcal{H}}
\renewcommand{\L}{\mathcal{L}}

\newcommand{\M}{\mathcal{M}}
\renewcommand{\O}{\mathcal{O}}

\renewcommand{\S}{\mathcal{S}}

\newcommand{\I}{\mathcal{I}}
\newcommand{\F}{\mathcal{F}}
\newcommand{\bigpi}{\scalebox{1.3}{$\pi$}}

\title{Classical and Quantum Query Complexity of Boolean Functions under Indefinite Causal Order}

\author{Alastair A.\ Abbott
\institute{Univ.\ Grenoble Alpes, Inria, 38000 Grenoble, France}
\and
Mehdi Mhalla
\institute{Univ.\ Grenoble Alpes, CNRS, Grenoble INP, LIG, 38000 Grenoble, France}
\and
Pierre Pocreau
\institute{Univ.\ Grenoble Alpes, Inria, 38000 Grenoble, France}
\institute{Univ.\ Grenoble Alpes, CNRS, Grenoble INP, LIG, 38000 Grenoble, France}
}

\begin{document}
\maketitle

\begin{abstract}
	Computational models typically assume that operations are applied in a fixed sequential order.
	In recent years several works have looked at relaxing this assumption, considering computations without any fixed causal structure and showing that such ``causally indefinite'' computations can provide advantages in various tasks.
	Recently, the quantum query complexity of Boolean functions has been used as a tool to probe their computational power in a standard complexity theoretic framework, but no separation in exact query complexity has thus-far been found.
	In this paper, we investigate this problem starting with the simpler and fully classical notion of deterministic query complexity of Boolean functions, and using classical-deterministic processes -- which may exhibit causal indefiniteness -- as a generalised computational framework.
	We first show that the standard polynomial and certificate lower bounds of deterministic query complexity also hold in such generalised models. 
	Then, we formulate a Boolean function for which causal indefiniteness permits a reduction in query complexity and show that this advantage can be amplified into a polynomial separation. 
	Finally, with the insights gained in the classical-deterministic setting, we give a Boolean function whose quantum query complexity is reduced by causally indefinite computations.
	
\end{abstract}

\section{Introduction}
\label{sec:introduction}

Computational models, both classical and quantum, typically assume that operations are applied in a fixed, sequential ``causal'' order. 
There has been growing interest, however, in relaxing such assumptions and exploring causally indefinite ways of processing information.
On the one hand, quantum mechanics allows the order in which quantum operations are applied to be controlled by other quantum systems and hence rendered indefinite, as in the quantum switch~\cite{chiribellaquantum2013}.
Such processes, which formally correspond to quantum supermaps~\cite{chiribella2009theoretical, chiribellaquantum2013} and can be represented in the process-matrix formalism~\cite{oreshkov2012quantum}, have been shown to provide advantages in several tasks in quantum information~\cite{chiribella2012perfect, araujo2014computational, bavaresco2021strict, guerin2016exponential}.
On the other hand, classical computational models exploiting closed timelike curves (CTCs) have been studied~\cite{brun2003computers,bacon2004quantum,aaronson2005guest, aaronson2009closed}, and it is known that causal indefiniteness in the process matrix framework is not a property unique to quantum processes.
Indeed, logically consistent classical (even deterministic) processes can be formulated that are incompatible with any causal explanation and can, for instance, exhibit maximally noncausal correlations while remaining paradox-free~\cite{baumeler2014maximal, BaumelerSpaceLogically2016}.

Sparked by this, there has been growing interest in studying more rigorously the computational power of these causally indefinite processes within a complexity theoretic framework.
One approach has been to connect quantum supermaps and their classical-deterministic counterpart, process functions~\cite{baumeler2016device}, to computation with linear CTCs in order to study their computational complexity~\cite{araujo2017quantum, baumeler2018computational}.
Quantum supermaps and process functions, however, are objects that most naturally describe how to compose various ``black-box'' operations (be it in a causally ordered, or indefinite, way).
It is perhaps most natural then to study their computational power within a query complexity framework. 
While various works have looked to quantify the number of queries needed to determine relative properties of multiple quantum black-box operations provided as input~\cite{araujo2014computational,taddei2021computational} or to compose them into another quantum operation as a function of the input operations~\cite{kristjansson2024exponentialseparationquantumquery}, these works fall outside the traditional framework of query complexity.

In a prior work~\cite{abbott2024quantumquerycomplexityboolean}, we extended the standard notion of quantum query complexity of Boolean functions~\cite{BUHRMAN2002Complexity,ambainis2018understanding} to general, potentially causally indefinite, quantum supermaps, showing that causal indefiniteness can indeed provide some advantages in this setting, with a reduction in error-probability obtained for computing some Boolean functions with two queries.
However, it remained an open question whether asymptotic separations, or advantages in the exact (rather than bounded-error) setting, can be obtained. 

In this paper, we propose to first take a step back and study the computational power of causal indefiniteness in classical processes using the query complexity model.
This simpler setting, in addition to being interesting in its own rights, can provide useful insights for the quantum case.
Working within the framework of process functions, we introduce a standard notion of deterministic query complexity of Boolean functions which generalises the well-known notion of decision tree complexity~\cite{BUHRMAN2002Complexity} to causally indefinite deterministic computations.

We prove two lower bounds on this quantity for process functions based on the polynomial representation~\cite{BUHRMAN2002Complexity} and the certificate complexity~\cite{BUHRMAN2002Complexity} of Boolean functions, and we show that causal indefiniteness can at best provide a quadratic advantage in query complexity.
We then look for explicit speedups obtainable using causal indefiniteness.
First, starting from the so-called ``Lugano'' (or ``AF/BW'') process function~\cite{baumeler2014maximal,BaumelerSpaceLogically2016} we construct an explicit Boolean function for which causal indefiniteness allows the deterministic query complexity to be reduced from four queries to three.
Then, we show that this constant separation can be amplified into a polynomial separation via a recursive construction, proving an asymptotic computational advantage from causal indefiniteness in the classical setting.

Finally, we build on these results in the classical-deterministic setting to prove that causal indefiniteness can also provide an equivalent constant advantage in quantum query complexity,
improving upon~\cite{abbott2024quantumquerycomplexityboolean} where only a reduced error for a fixed number of queries in the bounded-error model was obtained, and
answering several questions left open in~\cite{abbott2024quantumquerycomplexityboolean}.

\section{The deterministic query model of computation}
\label{sec:querymodel}

Recently, a generalised notion of quantum query complexity was used to compare the computational power of sequential and causally indefinite quantum computations~\cite{abbott2024quantumquerycomplexityboolean}.
In this query complexity framework, one aims to compute a Boolean function $f: \{0, 1\}^n \to \{0, 1\}$ using an algorithm that can only access the bits of the input $x \in \{0,1\}^n$ via queries to an oracle. 
This oracle can be taken to be either classical or quantum, and the computational complexity is quantified by counting the number of queries to the oracle.
This measure of complexity is a standard tool to compare different types of computations, including deterministic and probabilistic classical computations, or quantum computations~\cite{BUHRMAN2002Complexity}.
In particular, most of the known quantum advantages over classical algorithms can be understood in this framework~\cite{ambainis2018understanding}.

Whereas~\cite{abbott2024quantumquerycomplexityboolean} studied the quantum query complexity of causally indefinite computations, in this paper we first focus on the classical-deterministic query model of computation, where causal indefiniteness can also arise~\cite{BaumelerSpaceLogically2016}.
Here, the query oracle encoding $x$ is given as a function $O_x: \{1, \dots, n\} \to \{0, 1\}$ such that $O_x(i) = x_i$. 
A \emph{sequential} algorithm, which can be adaptive, can then be represented without loss of generality as a decision tree whose internal nodes are queries to the oracle $O_x$ and whose leaves represent the final outputs of the computation~\cite{BUHRMAN2002Complexity}. 
Such a decision tree is said to compute a Boolean function $f$ if, for all $x$, the evaluation of the tree outputs the value $f(x)$. 
The \emph{decision tree complexity} of $f$, which we will refer to as \emph{deterministic query complexity}, is defined as follows.
\begin{definition}[\cite{BUHRMAN2002Complexity}]
	\label{def:DTcomplexity}
  The \emph{deterministic query complexity} of a Boolean function $f$, denoted $D(f)$, is the minimum $T$ for which there exists a decision tree of depth $T$ computing $f$.
\end{definition}

This sequential notion of deterministic query complexity serves as a baseline to assess the computational power of causal indefiniteness in classical-deterministic scenarios. 

\section{Classical-deterministic processes}
\label{sec:classical-det}

To investigate the power of causal indefiniteness in the classical-deterministic query model, we will use the process function formalism~\cite{BaumelerSpaceLogically2016} to model more general causally indefinite classical computations.
This model encompasses standard sequential algorithms, and will allow us to compare them to causally indefinite computations in a unified framework.

The process function formalism describes how some deterministic operations (i.e., functions) $\vec{\mu} = (\mu_1, \dots, \mu_T)$, with $\mu_k : \I_k \to \O_k$, can be composed together in a logically consistent way~\cite{BaumelerSpaceLogically2016}.
This composition is described by the so-called process function $w : \mathcal{P} \times \bigtimes^T_{k=1} \mathcal{O}_k \to \bigtimes_{k=1}^T \mathcal{I}_k \times \F$, where $\mathcal{P}$ and $\F$ are, respectively, the domain and codomain of the function resulting from the composition. 
A process function with $T$ ``slots'' is thus defined in~\cite{BaumelerSpaceLogically2016} as the most general kind of transformation that maps the $T$ functions $\vec{\mu}$ into a function from $\mathcal{P}$ to $\F$.%
\footnote{We assume throughout that $\mathcal{P},\I_k,\O_k,\F$ are all nonempty. Note that if any of these are singletons then their role is trivial and they can be omitted from the specification of the process function. Note also that if $\mathcal{P}$ or $\F$ is a singleton, then the resulting function from $\mathcal{P}$ to $\F$ obtained from the process function is constant.} 
Process functions can be elegantly characterised through a fixed-point condition: for any choice of $a\in \mathcal{P}$ and operations $\vec{\mu}$, the process function has a unique fixed-point $\vec{i}$, as shown by the following proposition~\cite{baumeler2016device}.

\begin{restatable}[\cite{BaumelerSpaceLogically2016,baumeler2016device}]{proposition}{func}
  \label{def:charac}
      A function $w: \mathcal{P} \times \bigtimes^T_{k=1} \mathcal{O}_k \to \bigtimes_{k=1}^T \mathcal{I}_k \times \F$ is a \emph{process function} if and only if, for all $a \in \mathcal{P}$ and all $\vec{\mu} = (\mu_1, \dots, \mu_T)$ with $\mu_k = \I_k \to \O_k$,
\begin{equation}
  \label{eq:consistent}
  \exists!\ \Vec{i}: w\big(a, \Vec{\mu}(\Vec{i})\big) = (\Vec{i}, b),
\end{equation}
with $\vec{i} = (i_1, \dots, i_T) \in \bigtimes_{k=1}^T \I_k$ and $b \in \F$.
\end{restatable}

Whereas a standard query algorithm can be thought of as composing queries in a sequential order, process functions do not assume any such causal order, yet remain classical and deterministic.%
\footnote{One could also consider classical processes that are not required to be deterministic, for instance to study the bounded-error randomised query complexity of Boolean functions under indefinite causal order.
Note, however, that the set of such classical processes is not simply the convex hull of classical-deterministic ones; in particular, some classical processes exhibit a form of so-called fine-tuning \cite{BaumelerSpaceLogically2016}.}
We will thus refer to process functions as \emph{classical-deterministic processes}.

To study the comparative computational power of causally definite and indefinite processes in this framework, we begin by formalising the notion of a causally definite classical-deterministic process. 
To this end, we first introduce the \emph{induced functions} $w_k: \mathcal{P} \times \bigtimes_{k'=1}^T \O_{k'} \to \I_k$ and $w_F: \mathcal{P} \times \bigtimes_{k'=1}^T \O_{k'} \to \F$ defined as $w_k(a, \vec{o}) = w(a, \vec{o})_{\lvert k}$ and $w_F(a, \vec{o}) = w(a, \vec{o})_{\lvert {\F}}$ which specify the inputs given to the different operations $\mu_k$ or output to $\F$~\cite{baumeler2020equivalence}.
Note that, taken together, these uniquely define $w$.
We will also make use of the notion of a \emph{reduced process} $w^{|\mu_k}: \mathcal{P} \times \bigtimes^T_{\substack{k'=1 \\ k' \neq k}} \mathcal{O}_{k'} \to \bigtimes_{\substack{k'=1 \\ k' \neq k}}^T \mathcal{I}_{k'} \times \F$ obtained from $w$ by fixing the $k$th operation $\mu_k$~\cite{baumeler2020equivalence}.   
The reduced process $w^{|a}: \bigtimes^T_{k=1} \mathcal{O}_k \to \bigtimes_{k=1}^T \mathcal{I}_k \times \F$ is similarly obtained by fixing the value $a$ in the past $\mathcal{P}$.
Note that $w^{|\mu_k}$ is a $(T-1)$-slot classical-deterministic process, while $w^{|a}$ is a $T$-slot process with trivial input (i.e., $|\mathcal{P}|=1$), which we thus generally omit from its specification.
Finally, to compactly describe the composition of $w$ and operations $\vec{\mu}$ we will make use of a classical version of  the so-called ``link product''~\cite{chiribella2009theoretical}, written `$*$'. 
In particular, for any $a \in \mathcal{P}$ and operations $\vec{\mu}$, and writing $\vec{i}_a = (i_{a,1}, \dots, i_{a,T})$ the inputs of the corresponding fixed-point of $w$ we have $(w * \vec{\mu})(a) = w_F\big(\mu_1(i_{a,1}), \dots, \mu_T(i_{a,T})\big)$.
Full definitions of these notions are given in Appendix~\ref{appendix:induced}.

A classical-deterministic process is then said to be causally definite if it is compatible with the operations $\mu_k$ being applied in a sequential order. 
This sequential order may be fixed by $w$, but importantly it can also be dynamical: the choice of which operation to apply at step $k$ might depend on the choice of the previous operations applied (and their inputs), as well as the value of $a\in \mathcal{P}$ given to the process. We give a recursive definition inspired by the definition of multipartite causally separable quantum supermaps~\cite{wechs2019definition}.
In Appendix~\ref{appendix:definite} we show that these definitions are consistent if one represents classical-deterministic processes as (diagonal, deterministic) process matrices.
\begin{definition}
  \label{def:causallydef}
      A classical-deterministic process $w: \mathcal{P} \times \bigtimes^T_{k=1} \mathcal{O}_k \to \bigtimes_{k=1}^T \mathcal{I}_k \times \F$ is said to be \emph{causally definite} if and only if $T = 1$, or
      \begin{itemize}      
      \item if $|\mathcal{P}| = 1$, then there \emph{exists} $1 \leq k \leq T$ such that the induced function $w_k$ is constant, and such that for any operation $\mu_k:\I_k\to\O_k$, the reduced process $w^{|\mu_k}$ is itself causally definite;
        
        \item if $|\mathcal{P}| > 1$ (i.e., $\mathcal{P}$ is non-trivial), then for all $a \in \mathcal{P}$ the reduced process $w^{|a}$ is itself causally definite.
      \end{itemize}
\end{definition}

We end this section with an example of a \emph{causally indefinite} classical-deterministic process, the so-called Lugano (or AF/BW) process~\cite{baumeler2014maximal, BaumelerSpaceLogically2016}. 
This is a three-slot classical-deterministic process $w_{\text{Lugano}} : \{0, 1\}^3 \to \{0, 1\}^3$ defined by the induced functions $w_k$ (for $1\le k \le 3$)
 \begin{equation}
   \label{eq:lugano}
  w_k(o_1, o_2, o_3) = (1 \oplus o_{k \oplus_3 1}) o_{k \oplus_3 2},
\end{equation}
where $k \oplus_3 l := [(k + l -1) \bmod 3] + 1$, for any $l \in \mathbb{Z}$. 
Note that $w_{\text{Lugano}}$ is defined with trivial $\mathcal{P}$ and $\F$, which are thus omitted from its definition above.
One can show that $w_{\text{Lugano}}$ is logically consistent, as it verifies the unique fixed-point characterisation of Proposition~\ref{def:charac}~\cite{Baumeler2022unlimitednoncausal}. 
Furthermore, it is causally indefinite, as it is known to violate so-called causal inequalities~\cite{baumeler2016device}.
The Lugano process will be an important tool in our goal to prove query complexity advantages over causally definite computations in Sections~\ref{sec:exactSep} and \ref{sec:quantum}. 

\section{Generalised deterministic query complexity}
\label{sec:detQueryModelICO}

In order to study the query complexity power of classical-deterministic processes, we first need to extend the notion of deterministic query complexity introduced in Definition~\ref{def:DTcomplexity} beyond causally-definite computations -- which can be seen as evaluating queries to the oracle $O_x$ in a decision tree -- to the process function framework.
This generalisation is extremely natural, and simply interprets classical-deterministic processes as connecting the inputs and outputs of the different queries in a logically consistent way.
This generalisation is analogous to that proposed in~\cite{abbott2024quantumquerycomplexityboolean} for the quantum query complexity using general quantum supermaps.

\begin{definition}
  \label{def:dgen}
  The \emph{generalised deterministic query complexity} of a Boolean function $f$, denoted $D^\text{Gen}(f)$, is the minimum $T$ for which there exists a classical-deterministic process $w:  \{0, 1\}^T \to \{1, \dots, n\}^T  \times \{0, 1\}$ computing $f$ when making $T$ queries to $O_x$, i.e., such that for all $x \in \{0, 1\}^n$ and writing $\vec{O}_x = (\underbrace{O_x, \dots, O_x}_\textrm{T})$, we have $w *\vec{O}_x = f(x).$
\end{definition}

Note that, when computing a function $f$ with a process function $w$, the output $f(x)$ is thus obtained in the ``global future'' $\F=\{0,1\}$. We show that when computing a Boolean function $f$ in the query model, decision trees and causally definite classical-deterministic processes are equivalent, justifying the interpretation of $D^\text{Gen}(f)$ as a generalisation of $D(f)$.
\begin{restatable}[]{proposition}{equiv}
  \label{prop:equiv}
In the query model, there exists a decision tree of depth $T$ computing a Boolean function $f$ if and only if there exists a $T$-slot causally definite classical-deterministic process computing $f$.   
\end{restatable}

\begin{sproof}
  The implication in the forward direction is straightforward, as any decision tree of depth $T$ can be easily seen as defining a causally definite $T$-slot classical-deterministic process.

  For the other direction, while decision trees are always evaluated in a fixed sequential order, a causally definite classical-deterministic processes $w$ may be dynamical, meaning that the operation applied at each step can depend on the outputs of the previous operations. 
  However, this dynamicality is irrelevant when the $T$ operations are identical, and one can always find an equivalent non-dynamical classical-deterministic process whose action is the same as that of $w$ on $T$ queries to $O_x$. This non-dynamical process can then be seen as a depth $T$ decision tree. See Appendix~\ref{appendix:equiv} for the full proof.
\end{sproof}

It follows that for any Boolean function, $f$, $D^\text{Gen}(f) \leq D(f)$. 
By comparing the two complexities $D$ and $D^\textup{Gen}$ we can compare, on an equal footing, the standard deterministic query model of computation, which is intrinsically sequential, with the generalised one, allowing causally indefinite computations.
In particular, because the two complexities are defined relative to the same classical oracle, any separation between the two quantities provides proof of a computational advantage obtained from causal indefiniteness. 

\subsection{Degree and certificate lower bounds}
\label{sec:lower}

Before looking for such separations, we first show that some well-known lower bounds on $D$ can be generalised to $D^\text{Gen}$.
Such lower bounds on $D^\textup{Gen}$ are insightful, as they bound any possible advantages obtainable from causal indefiniteness and provide insight into which Boolean functions are candidates for obtaining an advantage.

The first lower bound we consider is based on the polynomial representation of Boolean functions. A multilinear polynomial $g : \mathbb{R}^n \to \mathbb{R}$, is said to represent $f$ if, for all $x \in \{0, 1\}^n$, $f(x) = g(x)$. 
It is known that any Boolean function $f$ has a unique multilinear polynomial representation, and that $\deg(f) \leq D(f) \leq \deg(f)^3$, where $\deg(f)$ denotes the degree of this representation~\cite{BUHRMAN2002Complexity}. 
We strengthen here this and show that $\deg(f)$ is also a lower bound on $D^\textup{Gen}$ (see Appendix~\ref{appendix:deglower} for the proof).
\begin{restatable}{proposition}{degBound}
  \label{prop:deglower}
  For any Boolean function $f$, $\deg(f) \leq D^\textup{Gen}(f).$
\end{restatable}
It follows that for any Boolean function $f$, $\deg(f) \leq D^{\text{Gen}}(f) \leq D(f) \leq \deg(f)^3,$ meaning that for functions such that $\deg(f) = D(f)$, there cannot be any computational advantage from causal indefiniteness. %

A second lower bound is obtained by looking at the certificate complexity of $f$~\cite{BUHRMAN2002Complexity}.
The certificate complexity quantifies the number of bits of $x$ one has to know, in the worst case, to compute $f(x)$.
Formally, a certificate for an input $x \in \{0, 1\}^n$ of $f$ is a set of indices $\mathnormal{I} \subseteq \{1,\dots,n\}$ such that, for all $y \in \{0, 1\}^n$, $y_{|\mathnormal{I}} = x_{|\mathnormal{I}}$ implies $f(x) = f(y)$, where $x_{|\mathnormal{I}}$ denotes the restriction of $x$ to the bits $x_i$ at indices $i\in\mathnormal{I}$. 
Writing $C(f,x)$ the smallest certificate for $x$, the certificate complexity of of $f$ is defined as $C(f) := \max_x C(f, x).$
It is known that for any function $f$, $C(f) \leq D(f) \leq C(f)^2$~\cite{BUHRMAN2002Complexity}. 
We show that $C(f)$ is also a lower bound of $D^\textup{Gen}$ (see Appendix~\ref{appendix:deglower}).
\begin{restatable}{proposition}{certBound}
  \label{prop:certlower}
  For any Boolean function $f$, $C(f) \leq D^\textup{Gen}(f).$
\end{restatable}

It follows that for any Boolean function $f$, $C(f) \leq D^{\text{Gen}}(f) \leq D(f) \leq C(f)^2,$ meaning that to obtain a separation between $D^\textup{Gen}$ and $D$, one also needs first a separation between $C$ and $D$ and that, furthermore, any separation will be at most quadratic. 

These two lower bounds provide criteria for finding Boolean functions that could exhibit an advantage from causal indefiniteness. 
In particular, it is known that almost all Boolean functions have full degree, meaning that the fraction of Boolean functions with $\deg(f) < D(f)$ is $o(1)$~\cite{BUHRMAN2002Complexity}.
Nevertheless, we show in the next section the existence of a Boolean function satisfying both $\deg(f) < D(f)$ and $C(f) < D(f)$, and for which we obtain a separation between $D(f)$ and $D^\textup{Gen}(f)$.

\section{A constant separation between $D^\textup{Gen}$ and $D$}
\label{sec:exactSep}

In order to demonstrate a separation between $D^\textup{Gen}$ and $D$ and hence a query complexity advantage from causal indefiniteness, we construct here a specific 6-bit Boolean function $f_{6c}$.
The function $f_{6c}$ can be represented as the degree-3 polynomial
\begin{equation}
  \label{eq:polyrep}
  f_{6c}(x_1, \dots, x_6) = x_4(1 \oplus x_2)x_3 \oplus x_5(1 \oplus x_3)x_1 \oplus x_6(1 \oplus x_1)x_2 \oplus x_1x_2x_3,
\end{equation}
where $\oplus$ denotes addition modulo 2.%
\footnote{One can easily check that Eq.~\eqref{eq:polyrep} is equivalent (on Boolean inputs) to the multilinear polynomial $g(x_1, \dots, x_6) = x_4(1 - x_2)x_3 + x_5(1 - x_3)x_1 + x_6(1 - x_1)x_2 + x_1x_2x_3$.}
Note that the structure of $f_{6c}$ is similar to the induced functions of the Lugano process~\eqref{eq:lugano}, with several monomials of the form $x_{k + 3}(1 \oplus x_{k \oplus_3 1})x_{k \oplus_3 2}$. 

It follows from Eq.~\eqref{eq:polyrep} that $\deg(f_{6c}) = 3$ and one can readily show that $C(f_{6c}) = 3$ as, for any $x$, the value of $f_{6c}(x)$ can be obtained from three bits of $x$. 
Furthermore, it is easy to compute $f_{6c}$ in a sequential fashion using four queries, and indeed we show that four queries are also necessary, so that $D(f_{6c}) = 4$. 
These properties (for which we give detailed proofs in Appendix~\ref{appendix:f6c}) mean that $f_{6c}$ satisfies the two conditions identified in Section~\ref{sec:lower} that make it a suitable candidate for a potential query complexity advantage from causal indefiniteness.

The following result shows that there is indeed a causally indefinite classical-deterministic process that computes $f_{6c}$ in three queries, yielding such an advantage.
\begin{restatable}{proposition}{fICO}
  \label{prop:fICO} $D^\textup{Gen}(f_{6c}) = 3$.  
\end{restatable}
\begin{sproof}
The classical-deterministic process that computes $f_{6c}$ in three queries is based on the Lugano process $w_\text{Lugano}: \{0, 1\}^3 \to \{0, 1\}^3$. 
To consider the action of $w_\text{Lugano}$ on the query oracle $O_x : \{1, \dots, 6\} \to \{0, 1\}$, we modify the process and consider the function $\overline{w}_\text{Lugano} : \{0, 1\}^3 \to \{1, \dots, 6\}^3 \times \{0, 1\}$, to which we have added a non-trivial future space $\F = \{0, 1\}$, and which is defined by the induced functions, for all $1 \leq k \leq 3$ and $(o_1, o_2, o_3) \in \{0, 1\}^3$,
\begin{equation}
  \overline{w}_k(o_1, o_2, o_3) = k + 3 \cdot w_k(o_1, o_2, o_3),
\end{equation}
where $w_k(o_1, o_2, o_3) = (1 \oplus o_{k \oplus_3 1}) o_{k \oplus_3 2}$ are the induced functions of $w_\text{Lugano}$ as in Eq.~\eqref{eq:lugano}, and
\begin{equation}
  \label{eq:output}
  \overline{w}_F(o_1, o_2, o_3) = o_1(1 \oplus o_2)o_3  \oplus o_2(1 \oplus o_3)o_1  \oplus  o_3(1 \oplus o_1)o_2  \oplus  o_1o_2o_3.
\end{equation}
The logical consistency of this function can be proven from that of the original Lugano process. 
To verify that it computes $f_{6c}$, one can check explicitly that $\overline{w}_{F} \big(O_x(i_{x,1}), O_x(i_{x,2}), O_x(i_{x,3}) \big) = f_{6c}(x)$, where $\vec{i}_x=(i_{x,1}, i_{x,2}, i_{x,3})$ is the unique fixed-point of $\overline{w}_{\text{Lugano}}$ acting on three copies of $O_x$.
The detailed proof is given in Appendix~\ref{appendix:f6c}.
\end{sproof}

These results exhibit an explicit $6$-bit function $f_{6c}$ with $D(f_{6c})=4$ and $D^\text{Gen}(f_{6c})=3$, showing that:
\begin{theorem}
  \label{th:ConstantSep}
  There exists a Boolean function $f$ for which $D^\textup{Gen}(f) < D(f)$.
\end{theorem}
This result proves that causal indefiniteness can provide a computational advantage in the standard framework of query complexity. 
In the next section, we show that this separation can serve as a stepping stone towards an asymptotic separation between  $D^\textup{Gen}$ and $D$ for a particular family of Boolean functions.

\section{A polynomial advantage from causal indefiniteness}
\label{sec:polyAdv}

In this section we show that a polynomial query complexity advantage is possible by constructing a family of Boolean functions based on a recursive iteration of the 6-bit function $f_{6c}$ defined in Eq.~\eqref{eq:polyrep}. 
Similar recursive constructions have previously used to amplify constant separations between exact quantum query complexity and deterministic query complexity into asymptotic ones (see, e.g.,~\cite{Bernstein97Quantum, ambainis2013superlinear, montanaro2015exact}). 
Given a Boolean function $f: \{0, 1\}^n \to \{0, 1\}$, the construction considers the recursively defined functions
\begin{equation}
  f^{(l+1)}(x_1, \dots, x_{n^{l+1}}) =  f \big(f^{(l)}(x_1, \dots, x_{n^{l}}), \dots, f^{(l)}(x_{(n-1)n^{l} + 1}\dots x_{n^{l+1}}) \big),
\end{equation}
with $f^{(1)}:=f$. The deterministic query complexity of these recursively defined functions is determined by the initial function $f$: if $D(f) = T$ then $D(f^{(l)}) = T^l$~\cite{montanaro2013compositiontheoremdecisiontree}.

By proving an analogous theorem for $D^\textup{Gen}$, one could amplify any constant separation between $D^\textup{Gen}$ and $D$ given by a function $f$ (such as $f_{6c}$) into a polynomial separation. 
Here we prove a slightly weaker result, which is nevertheless sufficient to obtain such a separation, and leave open the question of whether or not one can prove equality under composition for the generalised deterministic query complexity.

\begin{restatable}{theorem}{compo}
  \label{th:compo}
  For any Boolean function $f$ and for any depth $l > 1$,
  \begin{equation}
    D^\textup{Gen}(f) = T \implies D^\textup{Gen}(f^{(l)}) \leq T^l.
  \end{equation}
\end{restatable}

\begin{sproof}
  Starting with a classical-deterministic process $w : \bigtimes_{k=1}^{T} \O_k \to \bigtimes_{k=1}^{T} \I_k \times \F$  computing $f$ (which necessarily has $\O_k=\F=\{0,1\}$ and $\I_k=\{1,\dots,n\}$) and another classical-deterministic process $w^{(l)} : \bigtimes_{k=1}^{T^l} \O_k \to \bigtimes_{k=1}^{T^l} \I^{(l)}_k \times \F$ (with $\I^{(l)}_{k} = \{1, \dots, n^l\}$) computing $f^{(l)}$, the goal is to construct a process function $w^{(l+1)}$ computing $f^{(l+1)}$.
    
    We first construct a classical-deterministic process $\tilde{w}^{(l)}: \mathcal{P} \times  \bigtimes_{k=1}^{T^l} \O_k \to \bigtimes_{k=1}^{T^l} \I^{(l+1)}_k \times \F$, with $\mathcal{P}=\{1,\dots,n\}$, that computes $f^{(l)}$ on the $n$ consecutive sequences of length $n^l$ of a bitstring $x \in \{0, 1\}^{n^{l+1}}$. It is defined by its induced functions $\tilde{w}^{(l)}_k(a, o_1, \dots, o_{T^l}) = (a - 1)n^{l} + w^{(l)}_k(o_1, \dots, o_{T^l}),$ and $\tilde{w}^{(l)}_F(a, o_1, \dots, o_{T^l}) = w_F(o_1, \dots, o_{T^l})$, and its logical consistency can be seen to following from that of $w^{(l)}$. 
	In particular, this function is such that, for any $x \in \{0, 1\}^{n^{(l+1)}}$ and $a \in \mathcal{P}$, we have $(\tilde{w}^{(l)} * \vec{O}_x) (a) = f^{(l)}(x_{(a-1)n^l + 1}, \dots, x_{an^l})$.

    We then construct a classical-deterministic process $w^{(l+1)}$ from $\tilde{w}^{(l)}$ that computes $f^{(l+1)}$ in $T^{l+1}$ queries. This function is defined for any $T^{l+1}$ operations $\vec{\mu} = (\vec{\mu}_1, \dots, \vec{\mu}_T)$, with $\vec{\mu}_k = (\mu_{k,1}, \dots, \mu_{k,T^l})$ as
\begin{equation}
  w^{(l+1)} * \vec{\mu} = w * (\tilde{w}^{(l)} * \vec{\mu}_1, \dots, \tilde{w}^{(l)} * \vec{\mu}_T) \in \F.
\end{equation}
It follows from the properties of $w$ and $\tilde{w}^{(l)}$ that $w^{(l+1)}$ is logically consistent, and that it computes $f^{(l+1)}$ in $T+1$ queries. We provide a full proof in Appendix~\ref{appendix:compo}.
\end{sproof}

Theorem~\ref{th:compo} shows that we can amplify the separation obtained in Section~\ref{sec:exactSep} for the function $f_{6c}$, such that $D^\textup{Gen}(f_{6c}^{(l)}) \leq 3^l$ while $D(f_{6c}) = 4^l$. 
With that, the family of functions $\{f_{6c}^{(l)}\}_l$ gives us the following result.
\begin{theorem}
  \label{th:sep}
  There exists a family of Boolean functions $\{f_l\}_l$ with $D(f_l) \to \infty$ such that
  \begin{equation}
    D^\textup{Gen}(f_l) = O(D(f_l)^{0.792\dots}).
  \end{equation}
\end{theorem}

Theorem~\ref{th:sep} shows a polynomial separation between causally definite and indefinite computations obtained in a standard complexity-theoretic model. It provides a clear proof that causal indefiniteness can provide asymptotic advantages in query complexity, analogous to known separations between classical and quantum computational models~\cite{ambainis2018understanding}. 
Note moreover that in the setting we regard both sequential and causally indefinite computations are given access to the \emph{same} type of oracle $O_x$ -- something not true in the comparison of classical and quantum query complexity. 

Whether the separation we obtained is tight or not is an open question. As ${D^{\text{Gen}}(f_{6c}) = 3}$, proving a stronger version of Theorem~\ref{th:compo} with an equality instead of an upper-bound would imply a tight separation for the family of function $\{f^{(l)}_{6c}\}_l$. 
Nevertheless, as shown in Section~\ref{sec:lower}, $C(f)\le D^{\text{Gen}}(f) \le D(f) \le C(f)^2$, and hence the largest possible separation is quadratic.

\section{From a classical to quantum query complexity advantage}
\label{sec:quantum}

The asymptotic advantage in a fully classical setting shown in the previous section highlights the fact it is not only interesting to study causally indefinite computations in quantum settings.
Nonetheless, one of our primary motivations for studying this classical scenario was to gather insight into the quantum case.
In~\cite{abbott2024quantumquerycomplexityboolean}, where the quantum query complexity of Boolean functions was first studied, it was shown that indefinite causal order can reduce the probability of error in computing some Boolean functions with a fixed number of queries. 
However, no advantage in the number of queries to compute a function exactly was found, nor was any asymptotic separation -- either in the bounded-error or exact case -- observed.
Inspired by the stronger results obtained in Section~\ref{sec:exactSep} for the classical-deterministic setting, we revisit here the quantum setting and show that a constant separation in query complexity can also be obtained.

  More precisely, to show a separation in the quantum setting, where both sequential and causally indefinite processes can query superpositions of inputs, we consider another 6-bit Boolean function $f_{6q}$, obtained by modifying the function $f_{6c}$ used in Section~\ref{sec:exactSep}. 
We show that while causally definite quantum supermaps cannot compute this new function in three quantum queries, a modified quantum version of the Lugano process can compute it in three quantum queries. 
We finish by discussing whether or not such a constant advantage can be turned into an asymptotic with a similar method as that used in Section~\ref{sec:polyAdv} in the classical case.

\subsection{Quantum supermaps and the quantum query model of computation}
\label{sec:sup}

We begin by briefly introducing the framework of quantum supermaps and the quantum query model of computation, as well as some different mathematical tools and notation.

The process function formalism presented in Section~\ref{sec:classical-det} can be viewed as the classical, deterministic limit of a more general framework, that of quantum supermaps~\cite{BaumelerSpaceLogically2016}, which describes how quantum operations can be consistently composed and transformed~\cite{chiribellatransforming2008}.
A quantum supermap $\S$ with $T$ ``slots'' is then the most general transformation that maps any $T$ quantum channels $\M_k : \L(\H^{I_k}) \to \L(\H^{O_k})$ (i.e., completely positive (CP) and trace-preserving (TP) maps) into another quantum channel $\S(\M_1, \dots, \M_T) : \L(\H^P) \to \L(\H^F)$. 
To be operationally consistent, $\S$ is required to be $T$-linear, completely CP-preserving and TP-preserving~\cite{chiribellatransforming2008}. 
Using the Choi-Jamio\l{}kowski isomorphism~\cite{CHOI1975285, JAMIOLKOWSKI1972275}, a quantum channel can be represented as a positive semidefinite matrix $\mathsf{M}_k = \sum_{i,i'} \ketbra{i}{i'} \otimes \M_k(\ketbra{i}{i'}) \in \L(\H^{I_k} \otimes \H^{O_k})$ satisfying $\Tr_{O_k}( \mathsf{M}_k) = \id^{I_k}$, with $\Tr_{O_k}$ being the partial trace over the system $\H^{O_k}$. 
Writing $\H^{(IO)_k}=\H^{I_k O_k}=\H^{I}\otimes\H^{O}$, $[T] := \{1, \dots, T\}$, and $\H^{(IO)_{[T]}} = \bigotimes_{k \in [T]} \H^{(IO)_k}$, a quantum supermap can be also represented in the Choi picture as a ``process matrix''~\cite{oreshkov2012quantum}, a positive semidefinite matrix $W \in \L(\H^{P (IO)_{[T]} F})$ belonging to a specific subspace $\L^{\text{Gen}}$ (see Appendix~B of \cite{abbott2024quantumquerycomplexityboolean} for more details) and satisfying a normalisation constraint~\cite{araujowitnessing2015,wechs2019definition}. 
Quantum supermaps that compose the $\M_k$ in a fixed causal order represent standard quantum circuits or ``quantum combs''~\cite{chiribellaquantum2008}, and correspond to process matrices belonging to a subspace $\L^{\text{Seq}}\subset\L^{\text{Gen}}$ (see Appendix~C of~\cite{abbott2024quantumquerycomplexityboolean}).
We will generically denote general and sequential quantum supermaps as $\S^\textup{Gen}$ and $\S^\text{Seq}$, respectively.

The quantum query model of computation was extended from the standard, causally ordered setting~\cite{BUHRMAN2002Complexity,ambainis2018understanding} to general, potentially causally indefinite, quantum supermaps in~\cite{abbott2024quantumquerycomplexityboolean}.
This generalisation is analogous to that presented in Section~\ref{sec:querymodel} for the classical setting, with the crucial conceptual difference being that the classical oracles $O_x$ must be generalised to unitary quantum oracles $\widetilde{O}_x : \H^Q \to \H^{Q'}$, where $\H^{Q},\H^{Q'}$ are isomorphic $(n+1)$-dimensional Hilbert spaces\footnote{We distinguish here and below between isomorphic primed and unprimed spaces, indicating inputs and outputs, so that the corresponding Choi matrices can be unambiguously composed, in particular with process matrices (cf.\ Fig.~\ref{fig:Qsubroutine}).} such that, for any index $i \in \{0,\dots,n\}$, 
\begin{equation}
\label{eq:phase}
    \widetilde{O}_x\ket{i} = \begin{cases} (-1)^{x_i}\ket{i} & \text{ if } i \not = 0, \\
                                           \ket{i} & \text{ otherwise.} \end{cases}
\end{equation}
In particular, this allows the input bits to be queried in superposition. 
Note that the case $\ket{i} = \ket{0}$ is needed to ensure $\widetilde{O}_x$ is equivalent to the potentially more familiar oracle $\hat{O}_x\ket{j, b} = \ket{j, b \oplus x_j}$, with $1\leq j\leq n$, $b \in \{0, 1\}$~\cite{ambainisquantum2002,ambainis2018understanding}.
We denote the quantum channel associated to the unitary $\widetilde{O}_x$ as $\widetilde{\O}_x$, and its Choi matrix as $\widetilde{\mathsf{O}}_x$.

In this paper, we adopt the generalised definition of quantum query complexity formulated in~\cite{abbott2024quantumquerycomplexityboolean} in terms of quantum supermaps. 
This formulation is the quantum analogue of the generalised classical-deterministic query complexity (Definition~\ref{def:dgen}), and allows one to compare the quantum query complexity of causally definite and indefinite quantum supermaps on an equal footing.

\begin{definition}
  \label{def:Qgen}
  The \emph{generalised quantum query complexity} of a Boolean function $f$, $Q_E^{\text{Gen}}(f)$, is the minimum $T$ for which there exists a $T$-slot quantum supermap $\mathcal{S}^\text{Gen}$ such that for all $x \in \{0, 1\}$, measuring the qubit $\S^\textrm{Gen}(\underbrace{\widetilde{\O}_x, \dots, \widetilde{\O}_x}_\textrm{T copies}) = \rho_x\in\L(\H^F)$ in the computational basis gives $f(x)$ with probability $1$.%
  \footnote{We focus in this paper on the exact complexity $Q_E^{\text{Gen}}$, but the bounded-error complexity $Q_2^\text{Gen}$ is also of interest~\cite{abbott2024quantumquerycomplexityboolean}.}
\end{definition}

When restricted to quantum supermaps that are sequential, i.e., supermaps $\S^\text{Seq}$ whose process matrix representation belongs to the linear subspace $\L^\text{Seq} \subset \L^\text{Gen}$, one then recovers a notion of sequential query complexity $Q_E^\text{Seq}(f)$ which is equivalent to the standard notion of exact query complexity $Q_E(f)$ defined with quantum circuits. 
Note that it is known, for the particular task of computing Boolean functions in the query model, that the optimal causally definite computation can always be assumed to be sequential~\cite{abbott2024quantumquerycomplexityboolean}, even if more general causally definite computations are generally possible (e.g., with dynamical causal order)~\cite{wechs21}.

Some relations between the different query complexities introduced so far can be readily given. 
By definition we have, for any $f$, $Q_E^\text{Gen}(f) \leq Q_E(f)$ and furthermore, as the polynomial lower bound on $Q_E$ is also a lower bound of $Q_E^\textup{Gen}$ \cite{abbott2024quantumquerycomplexityboolean}, $\frac{\deg(f)}{2} \leq Q_E^{\text{Gen}}(f) \leq Q_E(f) \leq \deg(f)^3$. 
Because the classical-deterministic processes are subsets of quantum supermaps~\cite{BaumelerSpaceLogically2016} (see Appendix~\ref{appendix:FuncMat}), one also has $Q_E^\text{Gen}(f) \leq D^\text{Gen}(f)$. 
It then follows, by the polynomial lower bound on $Q_E^\textup{Gen}$ and since $D^\textup{Gen}(f) \leq \deg(f)^3$ (see Section~\ref{sec:lower}), that these two generalised query complexities are polynomially related. 
Finally, note that, in general, it is unclear exactly how $Q_E(f)$ and $D^\textup{Gen}(f)$ are related.

From these bounds, we see that only 
Boolean functions for which $Q_E > \deg(f)/2$ (i.e., the polynomial bound is not tight) are candidates to prove a separation between $Q_E^{\text{Gen}}$ and $Q_E$. 
In the next section, we show that such a separation is indeed possible. 

\subsection{A constant separation between $Q_E$ and $Q_E^\textup{Gen}$.}
\label{sec:SepQQgen}

The 6-bit function $f_{6c}$ studied in Section~\ref{sec:exactSep} can easily be seen to have $Q_E^\text{Gen}(f_{6c})=3$, since $\bar{w}_\text{Lugano}$ can itself be seen as a quantum supermap.
However, while $D(f_{6c})=4$ implies that $Q_E(f_{6c})\le 4$, it is not immediate that this bound is tight in the quantum setting. 
Indeed, using a semidefinite programming (SDP) formulation of the problem~\cite{barnum2003quantum} (cf.\ Appendix~\ref{sec:sdp}), one finds (up to numerical precision in the probability of obtaining the correct value of $f_{6c}(x)$) that $Q_E(f_{6c})=3$, and hence no advantage is obtained for this function in the quantum setting.
To prove a separation in terms of exact quantum query complexity, we instead modify the function $f_{6c}$ to define the new $6$-bit Boolean function
\begin{equation}
  \label{eq:f6q}
  f_{6q}(x_1, \dots, x_6) = f_{6c}(x_1 \oplus x_4, x_2 \oplus x_5, x_3 \oplus x_6, x_4, x_5, x_6).
\end{equation}
This function differs from $f_{6c}$ in that the references to $x_1$, $x_2$ and $x_3$ are replaced by $x_1 \oplus x_4$, $x_2 \oplus x_5$ and $x_3 \oplus x_6$ respectively. 
The extra difficulty of computing the parity of certain input bits at the same time as computing $f_{6c}$ itself will incur an extra cost in sequential algorithms, but, as we will see, can be performed at no extra cost by general supermaps.

Because $f_{6c}$ can be computed in four classical queries, it follows that $f_{6q}$ can be computed in four quantum queries, as the parity between two bits can be computed in a single quantum query~\cite{cleve1998quantum}. 
Indeed, defining for $1 \leq i,j \leq n$ a unitary $H_{i,j}: \H^Q \to \H^Q$ such that $H_{i,j}\ket{0} = \frac{\ket{i} + \ket{j}}{\sqrt{2}}$ and $H_{i,j}\ket{1} = \frac{\ket{i} - \ket{j}}{\sqrt{2}}$ (and which is completed elsewhere to be unitary), we have $H_{i,j} \widetilde{O}_x H_{i,j} \ket{0} = (-1)^{x_i}\ket*{x_i \oplus x_j}$. 
Using this subroutine to compute the parity between $x_1 \oplus x_4$, $x_2 \oplus x_5$ and $x_3 \oplus x_6$, one needs to query at most one of $x_4,x_5,x_6$ in order to compute $f_{6q}$ in the same manner that one computes $f_{6c}$ using four classical queries to $O_x$.

To prove that this sequential computation is optimal, we again use the SDP formulation of the problem due to~\cite{barnum2003quantum} that we detail in Appendix~\ref{sec:sdp}. 
This SDP allows one to obtain the minimum error $\varepsilon_T^\text{Seq}(f)$ with which one can compute a Boolean function $f$ using $T$ quantum queries. 
We find numerically that $\varepsilon_3^\text{Seq}(f_{6q})\approx 0.0207 > 0$, and hence obtain the following proposition (see Appendix~\ref{sec:sdp} for details).

\begin{restatable}{proposition}{num}
  \label{prop:Seqcomp}  
  $Q_E(f_{6q}) = 4$.  
\end{restatable}

In the next section we construct a causally indefinite quantum supermap, once again based on the Lugano process, which can compute the function $f_{6q}$ in three quantum queries, showing a quantum query complexity advantage from causal indefiniteness.

\subsubsection{A causally indefinite quantum supermap to compute $f_{6q}$}
\label{sec:lugQuantum}

As the function $f_{6q}$, defined by Eq.~\eqref{eq:f6q}, is constructed from $f_{6c}$, and because the latter can be computed with the Lugano process in three queries, a modified version of the Lugano process is a natural candidate to compute $f_{6q}$ in three queries.
Although the Lugano process $w_\text{Lugano}$ is a classical-deterministic process, it can be viewed as a quantum supermap whose process matrix is diagonal in a fixed pointer basis~\cite{BaumelerSpaceLogically2016,Araujo2017purification}. To compute $f_{6q}$ we will compose this supermap with some quantum subroutines to compute the parity between two given bits in one quantum query. 
Such subroutines can themselves be viewed as quantum supermaps taking a single quantum channel -- here, the oracle -- as input.
We prove the following proposition, whose full proof is given in Appendix~\ref{appendix:ICOcomp}.
\begin{restatable}[]{proposition}{Qsep}
  \label{prop:Qsep}
   $Q_E^{\textup{Gen}}(f_{6q}) = 3.$
\end{restatable}
\begin{sproof}
The classical-deterministic Lugano process $w_\text{Lugano}$ can be embedded in the process matrix formalism as a process matrix $W_\text{Lugano} \in \L(\H^{(IO)_{[3]}})$ where, for $1 \leq k \leq 3$, the $\H^{O_k}$ and $\H^{I_k}$ are 2-dimensional Hilbert spaces~\cite{BaumelerSpaceLogically2016,Araujo2017purification}.
Because the process is classical, we can further amend $W_\text{Lugano}$ with a future space $\H^F$ that keeps a copy of the classical values in the registers $\H^{O_1 O_2 O_3}$, defining
\begin{equation}
  \label{eq:genLugano}
\begin{split}
  \widetilde{W}_\text{Lugano} = \sum_{o_1, o_2, o_3 \in \{0, 1\}} & \ketbra{o_1, o_2, o_3}^{O_1 O_2 O_3}  \\
  & \otimes \ketbra{\overline{o}_2 o_3, \overline{o}_3 o_1, \overline{o}_1 o_2}^{I_1 I_2 I_3} \otimes \ketbra{o_1, o_2, o_3}^F,
\end{split}
\end{equation}
where $\overline{o}_k := (1 \oplus o_k)$. 
We proceed by composing this process with some quantum subroutines $G_k$, for $k \in \{1, 2, 3\}$, which adapt the qubit slots of $\widetilde{W}_\text{Lugano}$ for the oracle $\widetilde{O}_x$ which operates on a 7-dimensional input, and which can be used to compute the parity between different bits of $x$ in one quantum query. 
These subroutines, which correspond to 1-slot quantum supermaps whose additional outputs in the spaces $\H^{\alpha_k}$ are also sent to the global future, as well as their actions on the query oracle $\widetilde{O}_x$, are described in Figure~\ref{fig:Qsubroutine}.
\begin{figure}[t]
     \begin{center}         \includegraphics[width=1\textwidth]{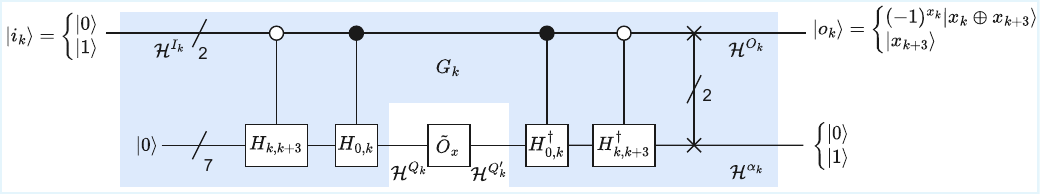}
 \end{center}
     \caption{Circuit diagram of the quantum subroutines $G_k$. Here, the Hilbert space $\H^{\alpha_k}$ has basis states $\ket{i}$ for $0 \leq i \leq 6$, and $H_{i,j}$ are unitaries satisfying $H_{i,j}\ket{0} = \frac{\ket{i} + \ket{j}}{\sqrt{2}}$ and $H_{i,j}\ket{1} = \frac{\ket{i} - \ket{j}}{\sqrt{2}}$ for $0 \leq i,j \leq 6$, $i \neq j$, and which can be completed arbitrarily on the other inputs.
	  The final operation is a swap operation $\textsc{swap}: \H^{O_k} \otimes \H^{\alpha_k} \to \H^{O_k} \otimes \H^{\alpha_k}$ such that for any $a,b\in\{0, 1\}$, $\textsc{swap} \ket{a}^{O_k} \ket{b}^{\alpha_k} = \ket{b}^{O_k} \ket{a}^{\alpha_k}$ and for any $c \in \{2, \dots, 6\}$, $\textsc{swap} \ket{a}^{O_k} \ket{c}^{\alpha_k} = \ket{a}^{O_k} \ket{c}^{\alpha_k}$.
    }
    \label{fig:Qsubroutine}
\end{figure}

The composition of $\widetilde{W}_\text{Lugano}$ with the $G_k$ in each slot defines a new three-slot quantum supermap with global future $\H^{F \alpha_1 \alpha_2 \alpha_3}$. 
Its process matrix is given by the link product as $\widetilde{W}_{f_{6q}} := \widetilde{W}_\text{Lugano} * G_1 * G_2 * G_3 \in \L(\H^{Q_1 Q_2 Q_3 Q'_1 Q'_2 Q'_3 F \alpha_1 \alpha_2 \alpha_3})$.
We then show that, for any value of $x$, when acting on three copies of the query oracle $\widetilde{O}_x$, we can obtain enough information to compute $f_{6q}$ by measuring the resulting state $\widetilde{W}_{f_{6q}} * (\widetilde{\mathsf{O}}_x \otimes \widetilde{\mathsf{O}}_x \otimes \widetilde{\mathsf{O}}_x) \in \L(\H^{F \alpha_1 \alpha_2 \alpha_3})$ in the computational basis.
\end{sproof}

Together, the propositions above show that the explicit 6-bit Boolean function $f_{6q}$ implies the following key result proving a separation in exact quantum query complexity obtained with a causally indefinite supermap. 
\begin{theorem}
  There exists a Boolean function $f$ for which $Q_E^{\text{Gen}}(f) < Q(f)$.
\end{theorem}

Given that the constant separation obtained in Section~\ref{sec:polyAdv} for the classical-deterministic setting could be amplified into an asymptotic polynomial separation, it is extremely natural to wonder whether a similar such separation can be obtained in the quantum case.
The recursive approach that was utilised in Section~\ref{sec:polyAdv} cannot, however, be directly applied to the constant quantum advantage obtained with $\widetilde{W}_{f_{6q}}$, despite this method previously being employed to exhibit quantum-over-classical query complexity advantages~\cite{ambainis2013superlinear}.
The main difficulty encountered is that $\widetilde{W}_{f_{6q}}$ does not describe a ``clean'' quantum supermap~\cite{aaronson2002quantum}, meaning that the output space of $\widetilde{W}_{f_{6q}} * \widetilde{\mathsf{O}}_x^{\otimes 3} \in \L(\H^{F \alpha_1 \alpha_2 \alpha_3})$ contains information about $f_{6q}(x)$, but also some additional information that cannot be uncomputed without additional queries to $\widetilde{O}_x$ -- which would negate the advantage obtained from causal indefiniteness.
While this is not a problem when computing $f_{6q}(x)$, it is an issue if one wants to coherently call this supermap recursively as a subroutine (see Appendix~\ref{appendix:ICOcomp} for further discussion).

\section{Discussion}
\label{sec:discussion}

In this paper we studied the potential of causally indefinite computations to provide advantages in query complexity, both in the classical-deterministic and quantum settings.
In the classical setting, we showed that classical-deterministic process functions can be used to obtain an asymptotic polynomial advantage in the query complexity of a Boolean function over causally ordered computations.
Although based on the well-studied Lugano process~\cite{baumeler2014maximal,BaumelerSpaceLogically2016}, our asymptotic separation required recursively composing together such processes to efficiently compute the function we constructed, an approach that has not previously been exploited when studying causally indefinite information processing.
Our results highlight that interest in causally indefinite computation is not restricted to the quantum setting where most advantages have been studied~\cite{baumeler2018computational}.
Here we focused exclusively on the deterministic setting, but it would be interesting to study also the potential advantages obtainable with more general (non-deterministic, stochastic) classical processes -- which are not simply described as the convex-hull of classical-deterministic processes~\cite{BaumelerSpaceLogically2016} -- in both the exact and bounded-error query complexity settings.

With the insights obtained from the study of the simpler classical-deterministic setting, we then proved a constant separation between the exact quantum query complexity $Q_E$ and its generalised version $Q_E^{\text{Gen}}$ introduced in~\cite{abbott2024quantumquerycomplexityboolean}, answering several questions that were left open therein. 
This reinforces the observation that causal indefiniteness can be a resource for computation even in more standard ``classical'' problems -- such as computing a Boolean function -- and not only in inherently quantum problems with no classical analogue, as most previous studies have investigated~\cite{araujo2014computational,taddei2021computational,kristjansson2024exponentialseparationquantumquery}.
Nevertheless, in contrast to the classical setting, we failed to amplify the observed separation with a recursive construction, and highlighted the need to uncompute the ``garbage'' ancillary qubits using additional queries~\cite{aaronson2002quantum} as a barrier to exploiting this technique.
We leave it as an open question to find an asymptotic separation in quantum query complexity in this setting.

Finally, we note that it would be interesting to study the query complexity for partial Boolean functions, i.e., with a promise on their inputs.
Indeed, even in standard quantum computational settings it is only possible to obtain super-polynomial quantum-over-classical advantages in query complexity for partial functions~\cite{barnum2003quantum}, as is the case, e.g., for the Deutsch-Jozsa problem~\cite{deutsch1992rapid,cleve1998quantum} or Simon's problem~\cite{SimonProblem}.
We showed that for total functions one can, at best, have a polynomial advantage from indefinite causal order in both classical-deterministic and quantum settings, but for partial functions larger advantages have not been ruled out.

\section*{Acknowledgments}

The authors acknowledge funding from the French National Research Agency through the Plan France 2030 projects ANR-22-CMAS-0001 and ANR-22-PETQ-0007, and the research grant ANR-22-CE47-0012. 
For the purpose of open access, the authors have applied a CCBY public copyright licence to any Author Accepted Manuscript (AAM) version arising from this submission.

\bibliographystyle{eptcs}
\bibliography{refs}

\appendix

\section{Classical-deterministic supermaps and generalised deterministic query complexity}

\subsection{Classical-deterministic processes}
\label{appendix:induced}

In this appendix, we give some more details regarding classical-deterministic processes, which were introduced in Section~\ref{sec:classical-det}. 
In particular, we provide detailed formal definitions of the notions of \emph{induced functions}, \emph{reduced process} and of the \emph{link product} between a classical-deterministic process and a set of functions. We start by recalling the fixed-point characterisation of process functions (which we concisely refer to as classical-deterministic processes).

\func*

A classical-deterministic process is uniquely defined by its so-called induced functions. 
These functions single out the input of a given operation inserted into a ``slot'' of the process.

\begin{definition}[Induced function \cite{baumeler2020equivalence}]
  Given a function $w: \mathcal{P} \times \bigtimes^T_{k=1} \mathcal{O}_k \to \bigtimes_{k=1}^T \mathcal{I}_k \times \F$, the \emph{induced functions} $w_k : \mathcal{P} \times \bigtimes^T_{k'=1} \mathcal{O}_{k'} \to  \mathcal{I}_k$, for $1 \leq k \leq T$, are defined such that for all $a \in \mathcal{P}$ and $\vec{o} = (o_1 \dots, o_T) \in \bigtimes_{k'=1}^T \O_{k'}$,
  \begin{equation}
    w_k(a, \vec{o}) = w(a, \vec{o})_{\lvert k} = (i_1, \dots, i_T, b)_{\lvert {k}} = i_k.
    \end{equation}
Similarly, for $\F$, we define $w_F : \mathcal{P} \times \bigtimes^T_{k'=1} \mathcal{O}_{k'} \to  \F$, such that for all $a \in \mathcal{P}$ and $\vec{o} = (o_1 \dots, o_T) \in \bigtimes_{k'=1}^T \O_{k'}$,
  \begin{equation}
     w_F(a, \vec{o}) = w(a, \vec{o})_{\lvert {\F}} = (i_1, \dots, i_T, b)_{\lvert {\F}} = b.
  \end{equation}
The induced functions $w_k$ are the restrictions of $w$ on the output spaces $\I_k$, while $w_F$ is its restriction on $\F$.
\end{definition}

Note that for $w$ to be logically consistent, it is necessary that for $1 \leq k \leq T$, the input $i_k$ received by the function $\mu_k$ is independent of its output $o_k$. 
If this is not the case, it opens the possibility for signalling ``back in time'' through $w$, from the output of $\mu_k$ to its input, which can lead to logical contradictions such as grandfather paradoxes~\cite{BaumelerSpaceLogically2016}. 
At the level of the induced functions, this means that for any $\vec{o}=(o_1, \dots, o_T) \in \bigtimes^T_{k'=1} \mathcal{O}_{k'}$, for any $a \in \mathcal{P}$ and any $1 \leq k \leq T$, $w_k(a, \vec{o})$ is independent of $o_k$, and we can thus unambiguously write $w_k(a,\vec{o}_{\backslash k})$, where $\vec{o}_{\backslash k} = (o_1, \dots, o_{k-1}, o_{k+1}, \dots, o_T)$.

The definition of classical-deterministic processes also implies that, by inserting an operation $\mu_k$ into the $k$th slot of a $T$-slot classical-deterministic process while leaving all the other slots empty, one obtains a $(T-1)$-slot classical-deterministic process $w^{| \mu_k}$, which we call a \emph{reduced process}. 
One can similarly define a reduced process $w^{|a}$ by fixing a value $a \in \mathcal{P}$. 
More formally, we have the following definition.

\begin{definition}[Reduced process \cite{baumeler2020equivalence}]
  Given a classical-deterministic process $w: \mathcal{P} \times \bigtimes^T_{k'=1} \mathcal{O}_{k'} \to \bigtimes_{k'=1}^T \mathcal{I}_{k'} \times \F$, for any $1 \leq k \leq T$ and any function $\mu_k: \mathcal{I}_k \to \mathcal{O}_k$, the \emph{reduced process} $w^{| \mu_k} : \mathcal{P} \times \bigtimes^T_{\substack{k'=1 \\ k' \neq k}} \mathcal{O}_{k'} \to \bigtimes_{\substack{k'=1 \\ k' \neq k}}^T \mathcal{I}_{k'} \times \F$ is the $(T-1)$-slot classical-deterministic process uniquely defined by the following induced functions, for all $k' \neq k$, for all $\vec{o}_{\backslash k} = (o_1, \dots, o_{k-1}, o_{k+1}, \dots, o_T) \in \bigtimes^T_{\substack{l=1 \\ l \neq k}} \mathcal{O}_{l}$ and for all $a \in \mathcal{P}$:
  \begin{align}
    w^{|\mu_k}_{k'}(a, \vec{o}_{\backslash k}) &= w_{k'}\left(a, \left[\vec{o}_{\backslash k}, \mu_{k}\big(w_k(a, \vec{o}_{\backslash k})\big)\right]\right) \text{,\quad and } \\
    w^{|\mu_k}_{F}(a, \vec{o}_{\backslash k}) &= w_F\left(a, \left[\vec{o}_{\backslash k}, \mu_k\big(w_k(a, \vec{o}_{\backslash k})\big)\right]\right),
  \end{align}
where we use the shorthand notation $[\vec{o}_{\backslash k}, \mu_k\big(w_k(a, \vec{o}_{\backslash k}) \big)] := \left(o_1, \dots, o_{k-1}, \mu_{k}\big(w_k(a, \vec{o}_{\backslash k})\big),o_{k+1},\dots,o_T\right)$.
  
  Similarly, one can define for any $a \in \mathcal{P}$ the reduced process $w^{|a}: \bigtimes^T_{k'=1} \mathcal{O}_{k'} \to \bigtimes_{k'=1}^T \mathcal{I}_{k'} \times \F$, defined by the following induced functions, for all $1 \leq k' \leq T$ and all $\vec{o} \in \bigtimes^T_{l=1} \mathcal{O}_{l}$:
  \begin{align}
    w_{k'}^{|a}(\vec{o}) &= w_{k'}(a, \vec{o}) \text{ and } \\
    w_F^{|a}(\vec{o}) &= w_F(a, \vec{o}).
    \end{align}    
\end{definition}

    We finish this appendix by defining a version of the so-called ``link product'' between a classical-deterministic process $w$ and some operations $\vec{\mu}$.\footnote{The link product was defined in \cite{chiribella2009theoretical} to describe the composition of quantum maps (cf.\ Appendix~\ref{appendix:FuncMat}). By extension, it can be used to describe the composition of different functions as well as that of stochastic channels. For the purpose of this paper, however, we only explicitly define the composition of classical-deterministic processes with a set of functions.} 
	This operation allows one to easily represent the resulting function from $\mathcal{P}$ to $\F$ obtained when $w$ acts on the operations $\vec{\mu}$.
  
\begin{definition}[Link product between a classical-deterministic process and a list of functions]
  \label{def:link}
    The \emph{link product}, denoted `$*$', between a classical-deterministic process $w: \mathcal{P} \times \bigtimes^T_{k=1} \mathcal{O}_k \to \bigtimes_{k=1}^T \mathcal{I}_k \times \F$ and a list of functions $\vec{\mu} = (\mu_1, \dots, \mu_T)$ is the function $w * \vec{\mu}: \mathcal{P} \to \F$ defined as, for any $a \in \mathcal{P}$ and writing $\vec{i}_a = (i_{a,1}, \dots, i_{a,T})$ the fixed-point of $w$ under the action of $a$ and the functions $\vec{\mu}=(\mu_1,\dots,\mu_T)$,
    \begin{equation}
    (w * \vec{\mu})(a) = w_F\big(\mu_1(i_{a,1}), \dots, \mu_T(i_{a,T})\big).
    \end{equation}    
  \end{definition}

\subsection{Causally definite classical-deterministic processes and decision trees}
\label{appendix:equiv}

In this appendix, we give a proof of Proposition~\ref{prop:equiv}, showing the equivalence of decision trees and causally definite classical-deterministic processes in the query model of computation.

\equiv*

\begin{proof}
We start by proving the necessary condition. 
Consider a binary decision tree of depth $T$ computing a Boolean function $f : \{0, 1\}^n \to \{0, 1\}$. 
Each internal node of the tree is labelled by an index $i \in \{1, \dots, n\}$ specifying a bit $x_i$ to query (i.e., the input $i$ for the oracle), and its leaves are labelled by binary values specifying the output of the computation. 
Note that we can assume, without loss of generality, that the tree is complete.
Every path from the root to a leaf of the tree can then be encoded as a binary string $(o_1, \dots, o_T) \in \{0, 1\}^T$ specifying which branch choice (i.e., left or right) to make at each internal node.
The decision tree is then in one-to-one correspondence with the functions $N_l: \{0, 1\}^{l-1} \to \{1, \dots, n\}$ for $1 \leq l \leq T+1$ which specify the label of each node at depth $l-1$ reached by following the path $(o_1,\dots,o_{l-1})$.
Note that $N_1(\lambda)$, where $\lambda$ is the empty string, is a constant (i.e., a function with no input) labelling the root (i.e., the node at depth 0), while $N_{T+1}(o_1,\dots,o_T)$ specifies the label of the leaf reached at the end of the corresponding path.  
Defining recursively the inputs for the oracle specified by the decision tree as $j_1:= N_1(\lambda)$ and $j_l:=N_l(x_{j_1},\dots,x_{j_{l-1}})$ for $l=2,\dots,T$, the output of the decision tree is thus $N_{T+1}(x_{j_1},\dots,x_{j_T})$.

The functions $N_l$ can then be used to define a classical-deterministic process $w_\text{Tree}: \{0, 1\}^T \to \{1, \dots, n\}^T \times \{0, 1\}$ via the induced functions $w_k(o_1, \dots, o_T) := N_k(o_1, \dots, o_{k-1})$, for all $1 \leq k \leq T$, and $w_F(o_1, \dots, o_T) := N_{T+1}(o_1, \dots, o_T)$. 
The logical consistency of this process follows from the structure of the decision tree, as for any set of operations $\vec{\mu} = (\mu_1, \dots, \mu_T)$, the vector $\vec{i} = (i_1, \dots, i_T)$ with $i_1 := N_1(\lambda)$ and, for $1 < k \leq T$, $i_k := N_k\big(\mu_1(i_1), \dots, \mu_{k-1}(i_{k-1})\big)$ will be the unique fixed-point of $w_{\text{Tree}}$ under $\vec{\mu}$. 
We thus see that, when $\mu_k=O_x$ for $k=1,\dots,T$ so that $\mu_k(i_k)=x_{i_k}$, we have precisely $w_F(\vec\mu(\vec{i}))=N_{T+1}(x_{i_1},\dots,x_{i_T})$ and hence $w$ computes $f$.

For the sufficient condition, note that while a decision tree is always evaluated in a fixed sequential order, a causally definite classical-deterministic process can have dynamical order, making the mapping more difficult.
Consider such process $w : \{0, 1\}^T \to \{1, \dots, n\}^T \times \{0, 1\}$ with $T$ slots that computes $f$.
It follows from Definition~\ref{def:causallydef} that there exists a $k_1 \in [T]$ such that the induced function $w_{k_1}$ outputs a constant $i_{k_1}$ and, for any operation $\mu_{k_1}$, the reduced function $w^{|\mu_{k_1}}$ is causally definite. 
We can thus define the function $N_1(\lambda) := w_{k_1}(\vec{o}') = i_{k_1}$ specifying the root of the decision tree, where $\vec{o}' \in \{0, 1\}^T$ can be taken to be an arbitrary string since $w_{k_1}$ is constant. 
Note that the reduced process $w^{|\mu_{k_1}}$ depends only on the value $\mu_{k_1}(i_{k_1})\in\{0,1\}$, so it is sufficient to consider the reduced processes $w^{|\mu^{[o_{1}]}_{k_1}}$ corresponding to the constant functions $\mu^{[o_{1}]}_{k_1}(\cdot)=o_{1}$ for $o_{1}\in\{0,1\}$.

Following the recursive nature of Definition~\ref{def:causallydef}, let us then introduce the recursively defined notation $w^{|\mu^{[o_1\dots o_{l}]}_{k_1\dots k_{l}}} := (w^{|\mu^{[o_1\dots o_{l-1}]}_{k_1\dots k_{l-1}}})^{|\mu^{[o_l]}_{k_l}}$ for $l>1$, where the $k_l$ are those such that each $(w^{|\mu^{[o_1\dots o_{l-1}]}_{k_1\dots k_{l-1}}})_{k_l}$ is constant.
We then define, for $1 < l \leq T$ and arbitrary $\vec{o}'\in\{0,1\}^{T-(l-1)}$, the functions

  \begin{equation}
    N_l(o_1, \dots o_{l-1}) := (w^{|\mu^{[o_1\dots o_{l-1}]}_{k_1\dots k_{l-1}}})_{k_l}(\vec{o}'),
  \end{equation}
  and
  \begin{equation}
    N_{T+1}(o_{1}, \dots, o_{T}) = (w^{|\mu^{[o_1\dots o_{T-1}]}_{k_1\dots k_{T-1}}})_F(o_{T}),
  \end{equation}
  where the reduced functions $(w^{|\mu^{[o_1\dots o_{l-1}]}_{k_1\dots k_{l-1}}})_{k_l}$ do not depend on $\vec{o}'$ since they are constant.
  
  To see that the decision tree defined by the functions $N_l$ for $1\le l \le T+1$ indeed computes $f$, consider the unique fixed-point $\vec{i}=(i_1,\dots,i_T)$ of $w$ acting on $T$ copies of $O_x$ and let $k_1,\dots,k_T$ be the order in which these operations are applied (which, except $k_1$, will in general depend on $x$).
  The output of the process is then 
  \begin{equation}
  	w_F(x_{i_1},\dots,x_{i_T}) = (w^{|\mu^{[x_{i_{k_1}}\dots x_{i_{k_{T-1}}}]}_{k_1\dots k_{T-1}}})_F (x_{i_{k_T}}) = N_{T+1}(x_{i_{k_1}},\dots,x_{i_{k_T}}),
  \end{equation}
  which is precisely the output of the decision tree thus defined.
    \end{proof}

\subsection{Degree and certificate lower bounds}
\label{appendix:deglower}

In this appendix, we prove the extended degree and certificate lower bounds for the generalised deterministic query complexity. 

\degBound*
\begin{proof}
  Let $f: \{0, 1\}^n \to \{0, 1\}$ be a Boolean function with $D^\text{Gen}(f) = T$. 
  Then, by definition, there exists a $T$-slot classical-deterministic process  $w:  \{0, 1\}^T \to \{1, \dots, n\}^T  \times \{0, 1\}$ such that, for all $x \in \{0, 1\}^n$, there exists a unique $\vec{i}_x$ satisfying 
\begin{equation}
    w_F\big(O_x(i_{x,1}), \dots, O_x(i_{x,T})\big) = f(x).
\end{equation}
We write the set of fixed-points for the strings $x$ that evaluate to $1$ as $\mathnormal{I}^{(1)}$;
note that different strings can have the same fixed-point, e.g., if the value of $f$ is independent of some bits in $x$.
Then, for each fixed-point $\vec{i}_x \in \mathnormal{I}^{(1)}$ we define a monomial of degree $T$ as

\begin{equation}
  R_{\vec{i}_x}(y) = \prod_{\mathclap{\substack{k=1 \\ O_x(i_{x,k}) = 1}}}^T y_{i_{x,k}} \quad \prod_{\mathclap{\substack{k=1 \\ O_x(i_{x,k}) = 0}}}^T (1 \oplus y_{i_{x,k}}).
\end{equation}
These monomials have the property that, for any bitstring $y\in\{0,1\}^T$, $R_{\vec{i}_x}(y) = 1$ implies that $\vec{i}_x = \vec{i}_y$ and $f(x) = f(y) = 1$. 
This means that for all $y$ there exist a unique monomial $R_{\vec{i}_x}$ such that $R_{\vec{i}_x}(y) = 1$ if $f(y) = 1$ and none if $f(y) = 0$. 
This implies that the polynomial 
\begin{equation}
  g(y) = \sum_{\vec{i}_{x} \in \mathnormal{I}^{[1]}} R_{\vec{i}_x}(y),
\end{equation}
of degree $T$ represents $f$, and therefore that $\deg(f) \leq D^\textup{Gen}(f)$.
\end{proof}

The second result generalises the certificate complexity bound to the generalised deterministic query complexity.

\certBound*
\begin{proof}
Let $f: \{0, 1\}^n \to \{0, 1\}$ a Boolean function with $D^\text{Gen}(f) = T$. 
Then, by definition, there exists a classical-deterministic process $w:  \{0, 1\}^T \to \{1, \dots, n\}^T  \times \{0, 1\}$ such that, for all $x$, there exists a unique $\vec{i}_x$ satisfying 
\begin{equation}
    w_F\big(O_x(i_{x,1}), \dots, O_x(i_{x,T})\big) = f(x).
\end{equation}
This implies that the set $\mathnormal{I}_x = \{ i_{x,1}, \dots, i_{x,T} \}$ is a certificate for $x$. 
Indeed, suppose we have a $y$ such that $y_{|\mathnormal{I}_x} = x_{|\mathnormal{I}_x}$, then $\vec{i}_x$ will also be a fixed-point under the queries to $O_y$, i.e.,
\begin{equation}
  w_F\big(O_x(i_{x,1}), \dots, O_x(i_{x,T})\big) = w_F\big(O_y(i_{x,1}), \dots, O_y(i_{x,T})\big) = f(x).
\end{equation}
By the unicity of fixed-points of classical-deterministic processes, we have $\vec{i}_y = \vec{i}_x$ and $f(x) = f(y)$. Therefore, for all $x$, $C(f,x) \leq |\mathnormal{I}_x| =T$ and hence $C(f) \leq T$.
\end{proof}

\subsection{Composition of the generalised deterministic query complexity}
\label{appendix:compo}

In this appendix, we prove the recursive composition theorem for the generalised deterministic query complexity of Boolean functions.

\compo*
\begin{proof}
 To prove this theorem, it is sufficient to show that from a classical-deterministic process computing $f$ in $T$ queries, one can build a classical-deterministic process that computes $f^{(l)}$ in $T^l$ queries. We will construct this process by induction.

Let $w : \bigtimes_k^{T} \O_k \to \bigtimes_k^{T} \I_k \times \F$ be a classical-deterministic process computing $f$, (which thus necessarily has $\O_k=F=\{0,1\}$ and $\I_k=\{1,\dots,n\}$) and $w^{(l)} : \bigtimes_k^{T^l} \O_k \to \bigtimes_k^{T^l} \I^{(l)}_k \times \F$ a classical-deterministic process computing $f^{(l)}$, where $\I^{(l)}_k = \{1, \dots, n^l\}$.
  
We begin by constructing a classical-deterministic process  $\tilde{w}^{(l)}$ that can compute $f^{(l)}$ on the $n$ consecutive sequences of length $n^l$ of a bitstring $x \in \{0, 1\}^{n^{l+1}}$. 
The construction is based on $w^{(l)}$ which is extended with a global past $\mathcal{P}$ in which it receives as input an offset $a \in \mathcal{P} = \{1, \dots, n\}$. 
It is this offset that determines on which sequence of $n^l$ bits the function $f^{(l)}$ is computed. 
Concretely the process is defined as $\tilde{w}^{(l)}: \mathcal{P} \times  \bigtimes_{k=1}^{T^l} \O_k \to \bigtimes_{k=1}^{T^l} \I^{(l+1)}_k \times \F$ through its induced functions
\begin{equation}
  \label{eq:tildew}
  \tilde{w}^{(l)}_k(a, o_1, \dots, o_{T^l}) = (a - 1)n^{l} + w^{(l)}_k(o_1, \dots, o_{T^l}),
\end{equation}
and
\begin{equation}
  \tilde{w}^{(l)}_F(a, o_1, \dots, o_{T^l}) = w_F(o_1, \dots, o_{T^l}).
\end{equation}

We now show that $\tilde{w}^{(l)}$ is logically consistent, and hence a valid classical-deterministic process. 
For any $a \in \mathcal{P}$ and for any set of functions $\{\tilde{\mu}_k\}_{k=1}^{T^l}$ with $\tilde{\mu}_k : \I_k^{(l+1)} \to \O_k$, we define the functions $\{\mu_{a, k}\}_{k=1}^{T^l}$ with $\mu_{a, k} : \I_k^{(l)} \to \O_k$ such that $\mu_{a,k}(i) = \tilde{\mu}_k\big((a-1)n^l + i \big)$. 
Now, denote $\vec{i} = (i_1, \dots, i_{T^l})$ the unique fixed-point of $w^{(l)}$ under the functions $\{\mu_{a, k}\}_{k=1}^{T^l}$, then the vector $\vec{j} = (j_1, \dots, j_{T^l})$ with $j_k = (a-1)n^l + i_k$ will be a fixed-point of $\tilde{w}^{(l)}$ under $a$ and the functions $\{\tilde{\mu}_k\}_{k=1}^{T^l}$:
\begin{align}
    \tilde{w}_k^{(l)}(a, j_1, \dots, j_{T^l}) &= (a-1)n^l + w_k^{(l)} \big(\tilde{\mu}_1(j_1), \dots, \tilde{\mu}_{T^{l}}(j_{T^l}) \big), \\
    &= (a-1)n^l + w_k^{(l)} \big(\mu_{a, 1}(i_1), \dots, \mu_{a, T^l}(i_{T^l}) \big), \\
    &= (a-1)n^l + i_k, \\
    &= j_k.
\end{align}
For the unicity of the fixed-point, suppose that $\vec{j'} = (j'_1, \dots, j'_{T^l})$ is a fixed-point of $\tilde{w}^{(l)}$ under $a \in \mathcal{P}$ and the functions $\{\tilde{\mu}_k\}_{k=1}^{T^l}$, then there exists an $\vec{i'} = (i'_1, \dots, i'_{T^l})$ such that 
\begin{align}
    \tilde{w}_k^{(l)}(a, j'_1, \dots, j'_{T^l}) &= (a-1)n^l + w_k^{(l)} \big(\tilde{\mu}_1(j'_1), \dots, \tilde{\mu}_{T^{l}}(j'_{T^l}) \big), \\
    &= (a-1)n^l + i'_k, \\
    &= j'_k.
\end{align}
Thus, $\vec{i'} = (i'_1, \dots, i'_{T^l})$ will be a fixed-point of $w^{(l)}$ under the functions  $\{\mu_{a, k}\}_{k=1}^{T^l}$, and because $w^{(l)}$ is logically consistent we have, by definition, $\vec{i'} = \vec{i}$, implying $\vec{j'} = \vec{j}$.

Furthermore, notice that for any $x \in \{0, 1\}^{n^{(l+1)}}$, with $T^l$ copies of the oracle $O_x$ and for any $a \in \mathcal{P}$
\begin{equation}
  \label{eq:shift}
    \tilde{w}^{(l)}_F\big(a, O_x(j_1), \dots O_x(j_{T^l})\big) = f^{(l)}(x_{(a-1)n^l + 1}, \dots, x_{an^l}),
\end{equation}
meaning that $\tilde{w}^{(l)}$ can compute $f$ on any of the $n$ consecutive sequence of $n^l$ bits of $x \in \{0, 1\}^{n^{l+1}}$.

The last step is to build a classical-deterministic process $w^{(l+1)}$ from $\tilde{w}^{(l)}$ that computes $f^{(l+1)}$ in $T^{l+1}$ queries. 
Consider a list of $T^{l+1}$ operations $\vec{\mu} = (\vec{\mu}_1, \dots, \vec{\mu}_T)$, where each $\vec{\mu}_k = (\vec{\mu}_{k,1}, \dots, \vec{\mu}_{k,T^l})$ and notice that for all $1 \leq k \leq T$, the function $\tilde{w}^{(l)} * \vec{\mu}_k$ is a function from $\mathcal{P}$ to $\F$, which are isomorphic, respectively, to the $\I_k$'s and $\O_k$'s. 
Let us then consider the value
\begin{equation}
  w * (\tilde{w}^{(l)} * \vec{\mu}_1, \dots, \tilde{w}^{(l)} * \vec{\mu}_T) \in \F.
\end{equation}
As such, the function $w^{(l+1)}: \bigtimes_{k}^{T^{l+1}} \O_k \to \bigtimes_{k}^{T^{l+1}} \O_k \times \F$ is defined for any $T^{l+1}$ operations $\vec{\mu}$ as
\begin{equation}
  w^{(l+1)} * \vec{\mu} = w * \big(\tilde{w}^{(l)} * \vec{\mu}_1, \dots, \tilde{w}^{(l)} * \vec{\mu}_T\big),
\end{equation}
and is a classical-deterministic process whose logical consistency follows from that of $w$ and $\tilde{w}^{(l)}$.

It follows from Eq.~\eqref{eq:shift} and because $w$ computes $f$ that, for any $x \in \{0, 1\}^{l+1}$ and by writing $\vec{O}_x = (\underbrace{O_x, \dots, O_x}_{T^l})$, we have
\begin{align}
  w * (\underbrace{\tilde{w}^{(l)} * \vec{O}_x, \dots, \tilde{w}^{(l)} * \vec{O}_x}_\textrm{T}) &= f\big(f^{(l)}(x_1, \dots, x_{n^{l}}), \dots, f^{(l)}(x_{(n-1)n^{l} + 1}\dots x_{n^l})\big) \\
    &= f^{(l+1)}(x),
\end{align}
thus completing the proof.
\end{proof}

\section{Properties of the Boolean function $f_{6c}$}
\label{appendix:f6c}
In this appendix, we prove different properties of the Boolean function $f_{6c}$, regarding its certificate complexity and its standard and generalised deterministic query complexity. In particular, we show that causal indefiniteness reduces the number of queries that are necessary to compute $f_{6c}$. The 6-bit Boolean function is represented by the following degree three polynomial.
\begin{align}
  f_{6c}(x_1, \dots, x_6) &= x_4(1 - x_2)x_3 + x_5(1 - x_3)x_1 + x_6(1 - x_1)x_2 + x_1x_2x_3, \label{eq:polyrep2a}\\
  &= x_4(1 \oplus x_2)x_3 \oplus x_5(1 \oplus x_3)x_1 \oplus x_6(1 \oplus x_1)x_2 \oplus x_1x_2x_3, \label{eq:polyrep2}
\end{align}
where $\oplus$ denotes addition modulo 2.
These two ways of writing the polynomial representing $f_{6c}$ are equivalent, and the multilinear polynomial representation~\eqref{eq:polyrep2a} is unique, so $\deg(f_{6c}) = 3$.
The truth table for $f_{6c}$ is also shown in Table~\ref{table:classicalf}.

\begin{table}[H]
      \caption{Simplified truth table of the Boolean function $f_{6c}$.}
    \label{table:classicalf}
    \centering
    \begin{tabular}{lll|c}
        $x_1$ & $x_2$ & $x_3$ & $f_{6c}(x)$ \\ \hline
        0 & 0 & 0 & 0 \\ 
        1 & 0 & 0 & $x_5$ \\ 
        0 & 1 & 0 & $x_6$ \\ 
        0 & 0 & 1 & $x_4$ \\ 
        1 & 1 & 0 & $x_5$ \\ 
        1 & 0 & 1 & $x_4$ \\ 
        0 & 1 & 1 & $x_6$ \\ 
        1 & 1 & 1 & 1 
    \end{tabular}
\end{table}

We start by proving the following result on the certificate complexity of $f_{6c}$.

\begin{proposition}
  \label{prop:deg&C} $C(f_{6c}) = 3.$
\end{proposition}
\begin{proof}
  The proof that $C(f_{6c}) = 3$ comes from observing Eq.~\eqref{eq:polyrep2} and verifying that, for any value of $x \in \{0, 1\}^6$, the value of $f_{6c}(x)$ depends only on three bits of the input.
  
  Indeed, when $x_1, x_2$ and $x_3$ are all equal to zero or to one, the value of $f_{6c}$ is independent of $x_4, x_5$ and $x_6$, meaning that $\{1, 2, 3\}$ is a certificate of size three. 
  When $x_2 = 0$ and $x_3 = 1$, the value of $f_{6c}$ is independent of $x_1, x_5$ and $x_6$, and $\{2, 3, 4\}$ is a size three certificate. 
  When $x_3 = 0$ and $x_1 = 1$,  $\{1, 3, 5\}$ is a certificate. 
  Lastly, when $x_1 = 0$ and $x_2 = 1,\{1, 2, 6\}$ is a certificate. 
  In each case, it is clear that one can't give a smaller certificate, as the output depends on each bit in the certificate.
  For every $x \in \{0, 1\}^6$ we thus have $C(f_{6c},x)=3$, and hence $C(f_{6c})=3$.
\end{proof}

We also prove the following result regarding the deterministic query complexity of $f_{6c}$.

\begin{proposition}
  \label{prop:Dg} $D(f_{6c}) = 4$.
\end{proposition}
\begin{proof}
Let us first note that one can indeed compute $f_{6c}$ with 4 queries.
	Suppose one first queries $x_1$, $x_2$ and $x_3$.
	If one obtains $x_1=x_2=x_3$, then the polynomial representation~\eqref{eq:polyrep2} reduces to $f_{6c}(x)=x_1x_2x_3$ and we are done.
	Otherwise, a single further query to either $x_4$, $x_5$ or $x_6$ is necessary, depending on whether $(x_2,x_3)=(0,1)$, $(x_1,x_3)=(1,0)$ or $(x_1,x_2)=(0,1)$, respectively.
	
	To see that one cannot compute $f_{6c}$ sequentially with 3 queries, let us first suppose that the first query is to $x_4$, $x_5$ or $x_6$.
	But then if $x_1=x_2=x_3$ three more queries will be needed to evaluate the monomial $x_1x_2x_3$, so one of these bits must be queried first if one is to use only three queries.
	Imagine then, without loss of generality, that the first query is to $x_1$.
	If $x_1=0$, then Eq.~\eqref{eq:polyrep2} reduces to $f_{6c}(x)=x_4(1\oplus x_2)x_3 \oplus x_6x_2$, and we see that, whichever of $x_2,x_3,x_4,x_6$ is chosen to be queried second, at least two further queries are needed in the worst case, giving a total of four queries, and hence no sequential adaptive algorithm can compute $f_{6c}$ in three queries for all $x$.
\end{proof}

These results show that $f_{6c}$ satisfies the two conditions identified in Section~\ref{sec:lower} that make a function a candidate for exhibiting a query complexity advantage from causal indefiniteness, namely that $\deg(f) < D(f)$ and $C(f) < D(f)$.
We now prove Proposition~\ref{prop:fICO}, showing that such a query complexity separation is obtained for $f_{6c}$.
\fICO*
\begin{proof}
To compute the Boolean function $f_{6c}$ in three queries, we take as a starting point the Lugano process $w_\text{Lugano}: \{0, 1\}^3 \to \{0, 1\}^3$ defined by its induced functions as in Eq.~\eqref{eq:lugano}.
In order to use the Lugano process to query several copies of the query oracle $O_x : \{1, \dots, 6\} \to \{0, 1\}$, we slightly modify the process so that the slots receive 6-dimensional inputs and we add a non-trivial future space $\F = \{0, 1\}$, giving a function $\overline{w}_\text{Lugano} : \{0, 1\}^3 \to \{1, \dots, 6\}^3 \times \{0,1\}$.
$\overline{w}_\text{Lugano}$ is defined by its induced functions which are, for all $(o_1, o_2, o_3) \in \{0, 1\}^3$ and any $1 \leq k \leq 3$,
\begin{equation}
  \overline{w}_k(o_1, o_2, o_3) = k + 3 \cdot w_k(o_1, o_2, o_3),
\end{equation}
where $w_k$ are the induced functions of the standard Lugano process as defined in Eq.~\eqref{eq:lugano}, and
\begin{equation}
  \label{eq:outputAppendix}
  \overline{w}_F(o_1, o_2, o_3) = o_1(1 \oplus o_2)o_3  \oplus o_2(1 \oplus o_3)o_1  \oplus  o_3(1 \oplus o_1)o_2 \oplus    o_1o_2o_3.
\end{equation}
Before showing that $\overline{w}_{\text{Lugano}}$ computes $f_{6c}$ in three queries, we first prove that it is logically consistent.

Consider three arbitrary functions $(\overline{\mu}_k)_{k\in \{1, 2, 3\}}$, with $\overline{\mu}_k : \{1, \dots, 6\} \to \{0, 1\}$, and define the three functions $(\mu_k)_{k \in \{1, 2, 3\}}$, with $\mu_k : \{0, 1\} \to \{0, 1\}$, such that $\mu_k(i) = \overline{\mu}_k(k + 3i)$. 
Because $w_\text{Lugano}$ is a classical-deterministic process, it has a unique fixed-point $(i_1, i_2, i_3) \in \{0, 1\}^3$ under the action of the $\mu_k$'s. 
Then $(j_1, j_2, j_3) \in \{1, 6\}^3$ with $j_k = k + 3i_k$ will be a fixed-point of $\overline{w}_{\text{Lugano}}$ under the actions of the $\overline{\mu}_k$'s.
Indeed,
\begin{align}
    \overline{w}_k\big(\overline{\mu}_1(j_1), \overline{\mu}_2(j_2), \overline{\mu}_3(j_3)\big) &= k + 3 \cdot w_k \big(\overline{\mu}_1(j_1), \overline{\mu}_2(j_2), \overline{\mu}_3(j_3) \big), \\
    &= k + 3 \cdot w_k \big(\mu_1(i_1), \mu_2(i_2), \mu_3(i_3) \big), \\
    &= k + 3 \cdot i_k, \\
    &= j_k.
\end{align}
To prove the unicity, suppose that $(j'_1, j'_2, j'_3) \in \{0, 1\}^6$ is another fixed-point of $\overline{w}_{\text{Lugano}}$ under the functions $(\overline{\mu}_k)_{k \in \{1, 2, 3\}}$. 
Then there exist $(i'_1, i'_2, i'_3) \in \{0, 1\}^3$ such that
\begin{align}
    \overline{w}_k\big(\overline{\mu}_1(j'_1), \overline{\mu}_2(j'_2), \overline{\mu}_3(j'_3)\big) &= k + 3 \cdot w_k \big(\overline{\mu}_1(j'_1), \overline{\mu}_2(j'_2), \overline{\mu}_3(j'_3) \big), \\
    &= k + 3 \cdot i'_k, \\
    &= j'_k.
\end{align}
This shows that $(i'_1, i'_2, i'_3)$ must also be a fixed-point of $w_{\text{Lugano}}$ under the functions $\{\mu_k\}_{k \in \{1, 2, 3\}}$, but because the Lugano is logically consistent this fixed-point is unique, and we have for $1 \leq k \leq 3$,  $i'_k = i_k$ and therefore $j'_k = j_k$.

Finally, to show that $\overline{w}_\text{Lugano}$ computes $f_{6c}$, we verify that for any $x \in \{0, 1\}^6$ and writing $\vec{O}_x = (\underbrace{O_x, O_x, O_x}_\textrm{3})$, we have
\begin{equation}
  \overline{w}_\text{Lugano} * \vec{O}_x = f(x).
\end{equation}
To see this, consider the first two columns of Table~\ref{table:lugCompute2} which show the unique fixed-points of $\overline{w}_\text{Lugano}$ for different values of $x$, under three copies of the query oracle $O_x$. 
These depend only on the values of $x_1, x_2$ and $x_3$ and one can verify that they are the unique fixed-points of $\overline{w}_\text{Lugano}$ by checking that, for each row and for any $1 \leq k \leq 3$, $\overline{w}_k\big(O_x(j_1), O_x(j_2), O_x(j_3)\big) = \overline{w}_k(x_{j_1}, x_{j_2}, x_{j_3}) = j_k$.

The last two columns of Table~\ref{table:lugCompute2} give the different outputs of the queries for the different fixed-points together with the value obtained in $\F$, which gives the output of the computation. 
We see that together, the first and last columns correspond exactly to the truth table of $f_{6c}$ in Table~\ref{table:classicalf}.
This shows that the classical-deterministic process $\overline{w}_\text{Lugano}$ computes $f_{6c}$ in three queries to $O_x$. Because $C(f_{6c}) = 3$, this implies that $D^\textup{Gen}(f_{6c}) = 3$.
\end{proof}

\begin{table}[!ht]
	\caption{Fixed-points and outputs of the classical-deterministic process $\overline{w}_\text{Lugano}$ on three copies of the query oracle $O_x$, for all $x \in \{0, 1\}^6$.
	In the table, $j_k = \overline{w}_k(o_1, o_2, o_3)$ and $o_k = O_x(j_k) =x_{j_k}$. The last column follows from Eq.~\eqref{eq:outputAppendix}.
	} \label{table:lugCompute2}
\centering
\begin{tabular}{lll|lll|lll|c}
$x_1$ & $x_2$ & $x_3$ & $j_1$ & $j_2$ & $j_3$ & $o_1$ & $o_2$ & $o_3$ & $\overline{w}_F(o_1, o_2, o_3)$ \\ \hline
0     & 0     & 0     & 1     & 2     & 3     & 0     & 0     & 0     & 0                    \\
1     & 0     & 0     & 1     & 5     & 3     & 1     & $x_5$ & 0     & $x_5$                \\
0     & 1     & 0     & 1     & 2     & 6     & 0     & 1     & $x_6$ & $x_6$                \\
0     & 0     & 1     & 4     & 2     & 3     & $x_4$ & 0     & 1     & $x_4$                \\
1     & 1     & 0     & 1     & 5     & 3     & 1     & $x_5$ & 0     & $x_5$                \\
1     & 0     & 1     & 4     & 2     & 3     & $x_4$ & 0     & 1     & $x_4$                \\
0     & 1     & 1     & 1     & 2     & 6     & 0     & 1     & $x_6$ & $x_6$                \\
1     & 1     & 1     & 1     & 2     & 3     & 1     & 1     & 1     & 1                    
\end{tabular}

\end{table}

\section{Process functions, process matrices and causal separability}
\label{appendix:SepAppendix}

\subsection{Process functions as process matrices}
\label{appendix:FuncMat}

We begin by discussing the restriction of the process matrix framework to classical-deterministic maps, and providing a short proof of the equivalence to process functions in this setting~\cite{BaumelerSpaceLogically2016,baumeler2016device}.
We fix a canonical computational ``classical'' basis, and consider an operation (a channel, CP map, or a process) to be classical if its Choi matrix is diagonal in that basis. 
For clarity, we will denote in this appendix the Choi matrices of such operations with a bar, e.g.\ as $\overline{\mathsf{M}}$.
Moreover, an operation is considered deterministic if it transforms pure classical states into pure classical states, in which case its Choi matrix contains only zeros and ones, and will generically be denoted in this appendix with a tilde, e.g.\ as $\widetilde{\mathsf{M}}$. 
A classical-deterministic channel, for example, corresponds to a function $\mu : \I \to \O$, and can be written in the Choi picture as the matrix $\widetilde{\mathsf{M}}_\mu = \sum_{i\in\I}\ketbra{i}^{I}\otimes\ketbra{\mu(i)}^{O}\in\L(\H^{IO})$.
In a slight abuse of language, we will sometimes identify a channel with its Choi matrix, the two being isomorphic.

Let us introduce the notation $\H^{(IO)_k}=\H^{I_k O_k}=\H^{I}\otimes\H^{O}$ and $\H^{(IO)_{[T]}} = \bigotimes_{k \in [T]} \H^{(IO)_k}$ with $[T] := \{1, \dots, T\}$, as well as $d_{O_k} = \dim(\H^{O_k})$, $d_{I_k} = \dim(\H^{I_k})$, $d_{O_{[T]}} = \prod_{k \in [T]} d_{O_k}$ and $d_{I_{[T]}} = \prod_{k \in [T]} d_{I_k}$.
We then begin by defining classical-deterministic process matrices, where without loss of generality we take the global past $\H^P$ and global future $\H^F$ to be trivial and therefore omit them~\cite{wechs2019definition}. 
We will make use of the link product \cite{chiribellaquantum2008, chiribella2009theoretical}, which allows the composition of quantum maps, potentially over a subset of their input/output systems, to be computed directly in the Choi picture.
The link product is defined for any matrices $\mathsf{M}^{XY}\in\L(\H^{XY})$ and $\mathsf{N}^{YZ} \in \L(\H^{YZ})$, as
$\mathsf{M}^{XY}*\mathsf{N}^{YZ} =\Tr_Y\big[(\mathsf{M}^{XY} \otimes \id^Z)^{\mathsf{T}_Y}(\id^X \otimes \mathsf{N}^{YZ})] \in \L(\H^{XZ})$, where $\cdot^{\mathsf{T}_Y}$ denotes the partial transpose over the Hilbert space $\H^Y$ with respect to the computational basis.

\begin{definition}[$T$-slot classical-deterministic process matrix]
  \label{def:classicalMat}
  A $T$-slot classical-deterministic process matrix is a diagonal positive semidefinite matrix $\widetilde{W} \in \L(\H^{(IO)_{[T]}})$ whose diagonal is composed only of zeros and ones, and such that for any $T$ classical-deterministic channels $(\widetilde{\mathsf{M}}_1, \dots, \widetilde{\mathsf{M}}_T)$ with $\widetilde{\mathsf{M}}_k \in \L(\H^{(IO)_k})$, we have $\widetilde{W} * (\widetilde{\mathsf{M}}_1 \otimes \dots \otimes \widetilde{\mathsf{M}}_T) = 1$.\footnote{This constraint of being well normalised on the set of classical-deterministic channels can also be framed in term of a normalisation constraint and a projection constraint, as it is done for general process matrix in Section~\ref{sec:sup} \cite{oreshkov2012quantum, araujowitnessing2015,wechs2019definition}}
\end{definition}

Note that it is sufficient to require that classical-deterministic process matrices be well normalised on the set of classical-deterministic channels, as for any $T$ \emph{quantum} channels $(\mathsf{M}_1, \dots, \mathsf{M}_T)$, with $\mathsf{M}_k \in \L(\H^{(IO)_k})$, we have $\widetilde{W} * (\mathsf{M}_1 \otimes \dots \otimes \mathsf{M}_T) = \widetilde{W} * (\overline{\mathsf{M}}_1 \otimes \dots \otimes \overline{\mathsf{M}}_T)$, with the $\overline{\mathsf{M}}_k = \bigpi_k(\mathsf{M}_k)$ being \emph{classical} channels, where $\bigpi_k$ is the projector onto the classical computational basis, i.e.\ $\bigpi_k(\mathsf{M}_k):= \sum_i \ketbra{i}\mathsf{M}_k\ketbra{i}$.
Since any classical channel can be written as a convex combination of \emph{classical-deterministic} channels, it follows from linearity that $\widetilde{W} * (\mathsf{M}_1 \otimes \dots \otimes \mathsf{M}_T) = 1$, and $\widetilde{W}$ is thus a valid process matrix.

A classical-deterministic process matrix $\widetilde{W}$ can also be seen as the Choi matrix of a classical-deterministic channel from $\L(\H^{O_{[T]}})$ to $\L(\H^{I_{[T]}})$. 
\begin{lemma}
  For any $T$-slot classical-deterministic process matrix $\widetilde{W} \in \L(\H^{(IO)_{[T]}})$ there exists a function $w: \bigtimes^T_{k=1} \mathcal{O}_k \to \bigtimes_{k=1}^T \mathcal{I}_k$ with $\O_k = \{0, \dots, d_{O_k}-1\}$ and $\I_k = \{0, \dots, d_{I_k}-1\}$, such that
  \begin{equation}
    \widetilde{W} = \sum_{\vec{k} \in \O_{[T]}} \ketbra*{\vec{k}}^{O_{[T]}} \otimes \ketbra*{w(\vec{k})}^{I_{[T]}},
  \end{equation}
where $\O_{[T]} = \bigtimes_{k \in [T]} \O_k$ and $\ketbra*{\vec{k}} = \ketbra{k_1 \dots k_T}$.
\end{lemma}

We note also that the characterisation of process functions given in Section~\ref{sec:classical-det} simplifies slightly in the case considered here, in which $\mathcal{P}$ and $\F$ are trivial and therefore omitted.
\begin{proposition}[Process function without $\mathcal{P}$ or $\F$ \cite{baumeler2016device}]
  \label{prop:unique}
      A function $w: \bigtimes^T_{k=1} \mathcal{O}_k \to \bigtimes_{k=1}^T \mathcal{I}_k$ is called a process function if and only if, for all $T$ functions $\vec{\mu} = (\mu_1, \dots, \mu_T)$ with $\mu_k = \I_k \to \O_k$,
      \begin{equation}
        \label{eq:uniq}
  \exists! \Vec{i}: w\big(\Vec{\mu}(\Vec{i})\big) = \Vec{i},
\end{equation}
with $\vec{i} = (i_1, \dots, i_T) \in \bigtimes_{k=1}^T \I_k$, and $\vec{\mu}(\vec{i}) = ( \mu_1(i_1), \dots, \mu_T(i_T))$
\end{proposition}

We can now show the equivalence between process functions and classical-deterministic process matrices, by noticing that the unique fixed-point condition of Proposition~\ref{prop:unique} is equivalent to the normalisation constraint of Definition~\ref{def:classicalMat}.
This result was first shown in~\cite{BaumelerSpaceLogically2016} and later in~\cite{baumeler2016device}, but we give an explicit self-contained proof here for completeness, as it will prove instructive for the following section.

\begin{proposition}
  A function $w: \bigtimes^T_{k=1} \mathcal{O}_k \to \bigtimes_{k=1}^T \mathcal{I}_k$, with the spaces $\O_k = \{0, \dots, d_{O_k}-1\}$ and $\I_k = \{0, \dots, {d_{I_k}-1}\}$ is a process function if and only if the matrix $\widetilde{W} \in \L(\H^{(IO)_{[T]}})$ defined as
  \begin{equation}
	\label{eq:processmatrixfromPF}
    \widetilde{W} = \sum_{\vec{k} \in \O_{[T]}} \ketbra*{\vec{k}}^{O_{[T]}} \otimes \ketbra*{w(\vec{k})}^{I_{[T]}}
  \end{equation}
  is a $T$-slot classical-deterministic process matrix.
\end{proposition}

\begin{proof}
  Let $w: \bigtimes^T_{k=1} \mathcal{O}_k \to \bigtimes_{k=1}^T \mathcal{I}_k$, with $\O_k = \{0, \dots, d_{O_k}-1\}$ and $\I_k = \{0, \dots, d_{I_k}-1\}$ be a process function, and define $\widetilde{W} = \sum_{\vec{k} \in \O_{[T]}} \ketbra*{\vec{k}}^{O_{[T]}} \otimes \ketbra*{w(\vec{k})}^{I_{[T]}}\in\L(\H^{(IO)_{[T]}})$. 
  By construction $\widetilde{W}$ is a diagonal positive semidefinite matrix, whose diagonal is composed only of zeros and ones. 
  Now consider any $T$ classical-deterministic channels $(\widetilde{\mathsf{M}}_{\mu_1}, \dots, \widetilde{\mathsf{M}}_{\mu_T})$ labelled with their corresponding functions $\mu_k : \O_k \to \I_k$, such that $\widetilde{\mathsf{M}}_{\mu_k} = \sum_{i \in \I_k} \ketbra*{i}^{I_k} \otimes \ketbra*{\mu_k(i)}^{O_k}$. 
  We have, writing $\I_{[T]} = \bigtimes_{k \in [T]} \I_k$,
  \begin{align}
    \widetilde{W} * (\widetilde{\mathsf{M}}_1 \otimes \dots \otimes \widetilde{\mathsf{M}}_T) &= \Tr\Big( \sum_{\vec{k} \in \O_{[T]}} \ketbra*{\vec{k}}^{O_{[T]}} \otimes \ketbra*{w(\vec{k})}^{I_{[T]}} \cdot \sum_{\vec{i} \in \I_{[T]}} \ketbra*{\vec{i}}^{I_{[T]}} \otimes \ketbra*{\vec{\mu}(\vec{i})}^{O_{[T]}} \Big) \\
    &= \Tr\Big(\sum_{ \substack{\vec{i} \in \I_{[T]} \\ w\big(\vec{\mu}(\vec{i})\big) = \vec{i}}} \ketbra*{\vec{\mu}(\vec{i})}^{O_{[T]}} \otimes \ketbra*{w\big(\vec{\mu}(\vec{i})\big)}^{I_{[T]}}\Big) \\
    &= 1,
  \end{align}
  where the last line follows from the existence of a unique $\vec{i}$ such that $w\big(\vec{\mu}(\vec{i}) \big) = \vec{i}$. 
  Conversely, if we assume that $\widetilde{W}$ is a valid classical-deterministic process matrix so that $\widetilde{W} * (\widetilde{\mathsf{M}}_1 \otimes \dots \otimes \widetilde{\mathsf{M}}_T) = 1$, we find in the same way that there must exist a unique $\vec{i}$ such that $w\big(\vec{\mu}(\vec{i})\big) = \vec{i}$.
\end{proof}

\subsection{Causal separability in classical-deterministic settings}
\label{appendix:definite}

In Definition~\ref{def:causallydef}, we proposed a definition of \emph{causally definite} classical-deterministic processes to describe processes that connect the various operations they take as input in a well-defined causal order. 
This definition is analogous in spirit to that of \emph{causal separability} for process matrices (or, equivalently, quantum supermaps), as defined in~\cite{wechs2019definition}.
The goal of this appendix is to demonstrate that our definition, formulated in the process function formalism, is indeed consistent with the more general notion of causal separability defined for process matrices when restricted to classical-deterministic process matrices.  

As in the discussion of classical process matrices in the previous section, we consider without loss of generality process matrices $W \in \L(\H^{(IO)_{[T]}})$ that are defined with trivial past and future space $\H^P$ and $\H^F$. 
We begin by recalling the definition of causal separability for process matrices.

\begin{definition}[Causal separability of process matrices \cite{wechs2019definition}]
 \label{def:causallysep}
  For $T=1$, any $T$-slot process matrix is causally separable. 
  For $T \geq 2$, a $T$-slot process matrix $W$ is said to be causally separable if and only if, for any extension $\H^{I'_{[T]}}$ of the slots' input spaces and any ancillary quantum state $\rho \in \L(\H^{I'_{[T]}})$, $W \otimes \rho$ can be decomposed as
  \begin{equation}
    \label{eq:decomp}
    W \otimes \rho = \sum_{k \in [T]} q_k W^{\rho}_{(k)},
  \end{equation}
with $q_k \geq 0$, $\sum_k q_k = 1$, and where for each $k$, $W^{\rho}_{(k)} \in \L(\H^{(II'O)_{[T]}})$ is a $T$-slot process matrix compatible with the $k$th operation acting first, and is such that for any CP map $\mathsf{M}_k \in \L(\H^{(II'O)_k})$, the conditional $(T-1)$-slot process matrix $(W^{\rho}_{(k)})_{|\mathsf{M}_k} := W^{\rho}_{(k)} * \mathsf{M}_k$ is itself causally separable (up to normalisation). 
\end{definition}

A process matrix that cannot be decomposed as in Definition~\ref{def:causallysep} is said to be causally nonseparable. 
In this definition, the extension with a state $\rho \in \L(\H^{I'_{[T]}})$ is important to rule out the possibility of ``activation'' of causal indefiniteness, wherein a process matrices that can be decomposed as in Eq.~\eqref{eq:decomp} when extended with any separable state can become causally nonseparable when extended with some entangled state~\cite{oreshkov2016causal,wechs2019definition}. 
The following proposition shows how this definition is simplified when the process we consider is classical, i.e., when its process matrix is diagonal in the computational basis. 

\begin{proposition}[Causal separability of classical process matrices]
  \label{prop:classicalSep}
  For $T=1$, any $T$-slot classical process matrix is causally separable. 
  For $T \geq 2$, a $T$-slot classical process matrix $\overline{W}$ is causally separable if and only if it can be decomposed as
  \begin{equation}
    \label{eq:simple}
    \overline{W} = \sum_{k \in [T]} q_k \overline{W}_{(k)},
  \end{equation}
  with $q_k \geq 0$, $\sum_k q_k = 1$, and where for each $k$, $\overline{W}_{(k)} \in \L(\H^{(IO)_{[T]}})$ is a classical process matrix compatible with the $k$th operation acting first, and is such that for any classical CP map $\overline{\mathsf{M}}_k \in \L(\H^{(IO)_k})$, the conditional $(T-1)$-slot classical process matrix $(\overline{W}_{(k)})_{|\overline{\mathsf{M}}_k} := \overline{W}_{(k)} * \overline{\mathsf{M}}_k$ is itself causally separable (up to normalisation).
\end{proposition}

\begin{proof}
Clearly, if $\overline{W}$ is causally separable then it has a decomposition of the form of Eq.~\eqref{eq:simple} with the desired properties since Definition~\ref{def:causallysep} includes the case of a trivial ancillary system and requires the recursive condition to hold for all CP maps $\mathsf{M}_k$, which includes classical CP maps $\overline{\mathsf{M}}_k$.
We hence focus on proving the converse statement.

We will prove that the simplified conditions above imply causal separability for classical process matrices inductively.
Clearly, for $T=1$ this is the case, since any $1$-slot process matrix, classical or not, is causally separable.

Let us then assume that any $(T-1)$-slot classical process matrix with a decomposition following Eq.~\eqref{eq:simple} with the required properties is causally separable.
Consider then a classical process matrix $\overline{W}$ that can be decomposed following Eq.~\eqref{eq:simple} and where for every $k\in [T]$ and any classical CP map $\overline{\mathsf{M}}_k \in \L(\H^{(IO)_k})$, $\overline{W}_{(k)} * \overline{\mathsf{M}}_k$ is a causally separable $(T-1)$-slot classical process matrix. 
To show that $\overline{W}$ is itself causally nonseparable, we will then show that for any extension $\H^{I'_{[T]}}$ of the parties' Hilbert spaces, any quantum state $\rho \in \L(\H^{I'_{[T]}})$, and any (not necessarily classical) CP map $\mathsf{M}_k \in \L(\H^{(II'O)_k})$, the extended process $\overline{W} \otimes \rho = \sum_{k \in [T]} q_k (\overline{W}_{(k)} \otimes \rho)$ indeed has the property that for any $k$ and $\mathsf{M}_k \in \L(\H^{(II'O)_k})$, $(\overline{W}_{(k)} \otimes \rho) * \mathsf{M}_k$ is causally separable.

Let us first note that because $\overline{W}_{(k)}$ is diagonal it satisfies
\begin{equation}
	\overline{W}_{(k)} = \sum_i \big(\Pi_i^{(IO)_k} \otimes \id^{(IO)_{[T] \backslash k}}\big) \overline{W}_{(k)} \big(\Pi_i^{(IO)_k} \otimes \id^{(IO)_{[T] \backslash k}}\big),
\end{equation}
where $\Pi_i^{(IO)_k} = \ketbra{i}{i}^{(IO)_k}$ and the $\ket{i}^{(IO)_k}$ are the computational basis vectors of $\H^{(IO)_k}$.
It then follows, using the definition of the link product, that
  \begin{align}
    (\overline{W}_{(k)} \otimes \rho) * \mathsf{M}_k &= \sum_i \Big(\big(\Pi_i^{(IO)_k} \otimes \id^{(IO)_{[T] \backslash k}}\big) \overline{W}_{(k)} \big(\Pi_i^{(IO)_k} \otimes \id^{(IO)_{[T] \backslash k}}\big) \otimes \rho \Big) * \mathsf{M}_k \\
    &= (\overline{W}_{(k)} \otimes \rho) * \sum_i \big(\Pi_i^{(IO)_k} \otimes \id^{I'_k}\big) \mathsf{M}_k  \big(\Pi_i^{(IO)_k} \otimes \id^{I'_k}\big). \label{eq:proj}
  \end{align}
  Writing $\mathsf{M}_k=\sum_{ij}\ketbra{i}{j}^{(IO)_k} \otimes \nu_{ij}^{I'_k}$, we note that since $\mathsf{M}_k$ is the Choi matrix of a CP map it is positive semidefinite, and hence the matrices $\nu_{ii}\in\L(\H^{I'_k})$ are also positive semidefinite.
  We then have, following Eq.~\eqref{eq:proj},
  \begin{align}
    (\overline{W}_{(k)} \otimes \rho) * \mathsf{M}_k &= (\overline{W}_{(k)} \otimes \rho) * \sum_i  \ketbra{i}^{(IO)_k} \otimes \nu_{ii}^{I'_k} \\
    &= \sum_i (\overline{W}_{(k)} * \ketbra{i}^{(IO)_k}) \otimes (\rho * \nu_{ii}^{I'_k}) \\
    &= \sum_i r_i (\overline{W}_{(k)} * \ketbra{i}^{(IO)_k}) \otimes \rho_i \\
    &= \sum_i \frac{r_{i}}{r} (\overline{W}_{(k)} * \overline{\mathsf{M}}_k^{(i)}) \otimes \rho_i, \label{eq:sep}
  \end{align}
  where $r_i:=\Tr(\rho * \nu_{ii}^{I'_k})$, $\rho_i:=\frac{1}{r_i}\rho * \nu_{ii}^{I'_k}\in\L(\H^{I'_{[T]\setminus k}})$ are ancillary states for the remaining $T-1$ parties, $r:=\sum_i r_i$, and the $\overline{\mathsf{M}}_k^{(i)}:= r\ketbra{i}^{(IO)_k}$ are classical (diagonal) CP maps.
  By the inductive assumption, for each $i$, the $(T-1)$-slot classical process matrix $(\overline{W}_{(k)})_{|\overline{\mathsf{M}}_k^{(i)}} := \overline{W}_{(k)} * \overline{\mathsf{M}}_k^{(i)}$ is causally separable and hence, by Definition~\ref{def:causallysep}, $(\overline{W}_{(k)})_{|\overline{\mathsf{M}}_k^{(i)}}\otimes \rho_i$ is also a causally separable $(T-1)$-slot process matrix.
  Since causal separability is preserved under convex combinations, we thus find conclude that $(\overline{W}_{(k)} \otimes \rho) * \mathsf{M}_k$ is itself causally separable, thereby completing the proof.
\end{proof}

For classical process matrices $\widetilde{W}$ that are moreover deterministic, causal separability can be characterised even more simply.
Indeed, it turns out that in this case one only needs to consider the reduced process matrices obtained by classical-deterministic actions $\widetilde{\mathsf{M}}_k$, and moreover it is sufficient to consider CPTP maps (i.e., channels) rather than arbitrary classical CP maps.
We hence find that, for such processes, causal separability is equivalent to the following simpler condition.

\begin{proposition}[Causally separable classical-deterministic process matrix]
  \label{prop:causallydefinite}
  For $T=1$, any $T$-slot clas\-sical-deterministic process matrix is causally separable. 
  For $T \geq 2$, a $T$-slot classical-deterministic process matrix $\widetilde{W}$ is causally separable if and only if there exists a $k$ for which $\widetilde{W}$ is compatible with the $k$th operation acting first, and such that for any classical-deterministic channel $\widetilde{\mathsf{M}}_k \in \L(\H^{(IO)_k})$, the $(T-1)$-slot classical-deterministic process matrix $\widetilde{W}_{|\widetilde{\mathsf{M}}_k} := \widetilde{W} * \widetilde{\mathsf{M}}_k$ is itself causally separable.  
\end{proposition}

\begin{proof}
  Clearly, if $\widetilde{W}$ is causally separable according to Proposition~\ref{prop:classicalSep}, then it also satisfies to conditions of this proposition, as classical-deterministic channels are a subset of classical CP maps, and the fact that $\widetilde{W}$ is deterministic means that it cannot be decomposed as a mixture of classical processes.
We hence focus on proving the converse statement.

We will prove that a classical-deterministic process matrix satisfying the conditions of the proposition is also causally separable according to Proposition~\ref{prop:classicalSep}.
Clearly, for $T=1$ this is the case, since any $1$-slot process matrix is causally separable.

Consider now that $\widetilde{W}$ is a classical-deterministic process matrix which is compatible with slot $k$ being applied first and for which, for any classical-deterministic channel $\widetilde{\mathsf{M}}_k \in \L(\H^{(IO)_k})$,  $\widetilde{W}_{|\widetilde{\mathsf{M}}_k}$ is causally separable.
We want to show that for any classical CP map, $\overline{\mathsf{M}}_k$, the conditional process matrix $\widetilde{W}*\overline{\mathsf{M}}_k$ is also causally separable according to Proposition~\ref{prop:classicalSep} (and hence also according to Definition~\ref{def:causallysep}).

Consider the maps $\{\widetilde{\mathsf{M}}_{k}^{[i,l]} = \ketbra{i}^{I_k} \otimes \ketbra{l}^{O_k} \}_{i,l}$, for $i \in \{0, \dots, d_{I_k}-1\}$ and $l \in \{0, \dots, d_{O_k}-1\}$, which form a basis for the space of classical-deterministic CP maps on $\L(\H^{(IO)_k})$. 
In particular,  any CP map $\overline{\mathsf{M}}_k$ can be decomposed as $\overline{\mathsf{M}}_k = \sum_{i,l} \mu_{i,l} \widetilde{\mathsf{M}}_{k}^{[i,l]}$ for some coefficients $\mu_{i,l} \geq 0$.

  Because $\widetilde{W}$ is deterministic, there exists a unique $i \in \{0, \dots, d_{I_k}-1\}$ (the input that slot $k$, which is applied first, receives) such that for any $l$, $\widetilde{W} * \widetilde{\mathsf{M}}_{k}^{[i,l]} \neq 0$, and it follows by linearity that
  \begin{equation}
    \widetilde{W} * \overline{\mathsf{M}}_k = \sum_{l} \mu_{i,l} \widetilde{W} * \widetilde{\mathsf{M}}_{k}^{[i,l]}.
  \end{equation}
For any $l$, $\sum_{j} \widetilde{\mathsf{M}}_{k}^{[j,l]}$ is a CPTP map and $\widetilde{W}* \sum_{j} \widetilde{\mathsf{M}}_{k}^{[j,l]} = \widetilde{W}*\widetilde{\mathsf{M}}_{k}^{[i,l]}$.
 It then follows that
  \begin{align}
    \widetilde{W} * \overline{\mathsf{M}}_k &= \sum_{l} \mu_{i,l} \widetilde{W} * \widetilde{\mathsf{M}}_{k}^{[i,l]} \\
      &= \sum_{l} \mu_{i,l} \widetilde{W} * \sum_{j} \widetilde{\mathsf{M}}_{k}^{[j,l]},
  \end{align}
  and, by assumption, $\widetilde{W} * \sum_{j} \widetilde{\mathsf{M}}_{k}^{[j,l]}$ is a causally separable $(T-1)$-slot classical-deterministic process matrix. 
  The conditional process matrix $\widetilde{W} * \overline{\mathsf{M}}_k$ is therefore, by linearity, also causally separable (up to normalisation), which concludes the proof.
\end{proof}

Following the equivalence of classical-deterministic process matrices and process functions (see Appendix~\ref{appendix:FuncMat}), we obtain the following corollary.

\begin{corollary}
	A classical-deterministic process $w: \bigtimes^T_{k=1} \mathcal{O}_k \to \bigtimes_{k=1}^T \mathcal{I}_k$ is causally definite if and only if the corresponding classical-deterministic process matrix $\widetilde{W} \in \L(\H^{(IO)_{[T]}})$ defined as in Eq.~\eqref{eq:processmatrixfromPF} is causally separable.
\end{corollary}

This proves the consistency of our proposed definition of causally definite process functions with the established notion of causal separability for process matrices.
We reiterate that this equivalence holds also in the case where $\mathcal{P}$ and $\F$ are nontrivial, simply be reinterpreting these as additional slots with trivial input and output spaces, respectively.

\section{Properties of the function $f_{6q}$}
\subsection{A numerical method to compute the quantum query complexity of $f_{6q}$}
\label{sec:sdp}
In this appendix, we give some details on the semidefinite program (SDP) mentioned in Section~\ref{sec:SepQQgen} to show that the Boolean function $f_{6q}$ cannot be computed sequentially with less than four quantum queries. The SDP is due to \cite{barnum2003quantum} and was used in \cite{montanaro2015exact} to obtain the minimum error $\varepsilon_{T}^\text{Seq}(f)$ with which one can compute a Boolean function $f$ in $T$ queries to the quantum oracle $\widetilde{O}_x$, i.e., for which there exists a $T$-slot supermap $\S^\text{Seq}$ such that for all $x$, measuring the qubit state $\S^\text{Seq}(\widetilde{\O}_x, \dots, \widetilde{\O}_x) = \rho_x$ in the computational basis gives the outcome $f(x)$ with probability greater than $1 - \varepsilon_{T}^\text{Seq}(f)$. 
The SDP is formulated as follows, where $\circ$ is the Hadamard (entry-wise) matrix product.
\begin{theorem}[\cite{barnum2003quantum}]
  \label{th:sdp}
  The minimum bounded error $\varepsilon_{T}^\text{Seq}(f)$ for which there exists a sequential supermap $\mathcal{S}^\textup{Seq}$ that computes $f: \{0, 1\}^n \to \{0, 1\}$ in $T$ quantum queries is given by the SDP 
  \begin{align}
    \varepsilon_{T}^\text{Seq}(f) = \min_{\varepsilon,\{M_i^{(j)}\}_{i,j},\Gamma_0,\ \Gamma_1} \quad \varepsilon \\
   \text{s.t.}\quad \sum_{i=0}^n M_i^{(0)} &= E_0, \\
    \sum_{i=0}^n M_i^{(j)} &= \sum_{i=0}^n E_i \circ M_i^{(j-1)}, \text{ for } 1 \leq j \leq T-1, \\
    \Gamma_0 + \Gamma_1 &= \sum_{i=0}^{n} E_i \circ M_i^{(T-1)}, \\
    F_0 \circ \Gamma_0 &= (1-\varepsilon)F_0, \\
    F_1 \circ \Gamma_1 &= (1-\varepsilon)F_1,
  \end{align}
  where the $\{M_i^{(j)}\}_{i,j}$ (for $0 \leq i \leq n$ and $0 \leq j \leq T-1$) and $\Gamma_0,\Gamma_1$ are all $2^n$-dimensional real symmetric positive semidefinite matrices, $E_0$ is the constant 1 matrix and for $1\le i \le n$ the $E_i$ are defined such that $\bra{x}E_i\ket{y} = (-1)^{x_i + y_i}$, and where $F_0$ and $F_1$ are diagonal matrices such that $\bra{x} F_z \ket{x} = 1$ if and only if $f(x) = z$, otherwise $\bra{x} F_z \ket{x} = 0$.
\end{theorem}

This method was used in \cite{montanaro2015exact} to compute the sequential quantum query complexity of symmetric Boolean functions up to 6 bits in order to prove some separations between $D$ and $Q_E$.
However, because $f_{6q}$ is not symmetric, the computation was not performed for this specific function. 
Solving numerically the SDP, we obtain the following result.
\num*
\begin{proof}
By fixing the number of queries to $T=3$, and solving the SDP of Theorem~\ref{th:sdp} we obtain $\varepsilon_3^\text{Seq}(f_{6q}) = 0.0207$. 
This result is subject to numerical imprecisions, but with both the primal and dual converging to the same values and with constraints that are violated only up to a magnitude of $10^{-8}$, this value of $\varepsilon_3^\text{Seq}$ can be judged sufficiently far from 0 to assert with confidence that $f_{6q}$ cannot be computed in three queries.\footnote{To obtain a rigorous lower-bound on $\varepsilon_3^\text{Seq}$, one could extract a certificate from the dual SDP of Theorem~\ref{th:sdp}, in a method similar to that developed in \cite{bavaresco2021strict}.} 
The proof is completed by recalling that $f_{6q}$ can be computed in four quantum queries as described in Section~\ref{sec:SepQQgen}.
\end{proof}

Our code is freely accessible on Github.\footnote{\url{https://github.com/pierrepocreau/Classical-QuantumQueryComplexity_ICO}} with the SDP of Theorem~\ref{th:sdp} being formulated using the package Yalmip~\cite{yalmip} and solved using Mosek~\cite{mosek}.

\subsection{A causally indefinite quantum supermap computing $f_{6q}$}
\label{appendix:ICOcomp}
We finish by providing a more detailed proof the a causally indefinite quantum supermap that can compute the Boolean function $f_{6q}$ in three quantum queries. 
The truth table for $f_{6q}$ is shown in Table~\ref{table:classicalf}, and will be useful for the following proposition.
\begin{table}[ht]
      \caption{Simplified truth table of the Boolean function $f_{6q}$.}
    \label{table:quantum}
    \centering
    \begin{tabular}{c|c|c|c}
        $x_1 \oplus x_4$ & $x_2 \oplus x_5$ & $x_3 \oplus x_6$ & $f_{6q}(x_1, \dots, x_6)$ \\ \hline
        0 & 0 & 0 & 0 \\ 
        1 & 0 & 0 & $x_5$ \\ 
        0 & 1 & 0 & $x_6$ \\ 
        0 & 0 & 1 & $x_4$ \\ 
        1 & 1 & 0 & $x_5$ \\ 
        1 & 0 & 1 & $x_4$ \\ 
        0 & 1 & 1 & $x_6$ \\ 
        1 & 1 & 1 & 1 
    \end{tabular}

\end{table}

\Qsep*
\begin{proof}
Recall that the classical-deterministic Lugano process is defined as a function $w_{\text{Lugano}}$ mapping the binary output sets $\O_1 \times \O_2 \times \O_3$ to the binary input sets $\I_1 \times \I_2 \times \I_3$ such that $w_k(o_1, o_2, o_3) = (1 \oplus o_{k \oplus_3 1}) o_{k \oplus_3 2}$ (see Section~\ref{sec:classical-det}). 
This classical-deterministic process can be described a process matrix $W_\text{Lugano} \in \L(\H^{(IO)_{[3]}})$ where, for $1 \leq k \leq 3$, the $\H^{O_k}$ and $\H^{I_k}$ are two-dimensional Hilbert spaces,
\begin{equation}
  W_\text{Lugano} = \sum_{o_1, o_2, o_3 \in \{0, 1\}} \ketbra{o_1, o_2, o_3}^{O_1 O_2 O_3} \otimes \ketbra{\overline{o}_2 o_3, \overline{o}_3 o_1, \overline{o}_1 o_2}^{I_1 I_2 I_3},
\end{equation}
with $\overline{o}_k := (1 \oplus o_k)$. 
Note that any classical-deterministic process can be embedded in such a way \cite{BaumelerSpaceLogically2016, Araujo2017purification}; see also Appendix~\ref{appendix:FuncMat}. 
Because this process is classical, we can extend it with a (8-dimensional) future space $\L(\H^F)$ which keeps a copy of the classical values in the output registers $\L(\H^{O_1 O_2 O_3})$, defining the process

\begin{equation}
  \label{eq:genLuganoAppendix}
\begin{split}
  \widetilde{W}_\text{Lugano} = \sum_{o_1, o_2, o_3 \in \{0, 1\}} & \ketbra{o_1, o_2, o_3}^{O_1 O_2 O_3} \\
  &  \otimes \ketbra{\overline{o}_2 o_3, \overline{o}_3 o_1, \overline{o}_1 o_2}^{I_1 I_2 I_3} \otimes \ketbra{o_1, o_2, o_3}^F.
\end{split}
\end{equation}

We now define some quantum subroutines that can be used to compute the parity between different bits in one quantum query. 
These subroutines, indexed by $k \in \{1, 2, 3\}$, are formally one-slot quantum supermaps with corresponding Choi representations (i.e., process matrices) $G_k \in \L(\H^{I_k Q_k Q'_k O_k \alpha_k})$, and take a quantum channel from $\L(\H^{Q_k})$ to $\L(\H^{Q'_k})$ as input, with both  $\H^{Q_k}$ and $\H^{Q'_k}$ being Hilbert spaces of dimension 7, and where the $\H^{\alpha_k}$ are also a Hilbert spaces of dimension $7$. 
We show a circuit description of these supermaps in Figure~\ref{fig:Qsubroutine2}.
\begin{figure}[t]
     \centering
 \includegraphics[width=1\textwidth]{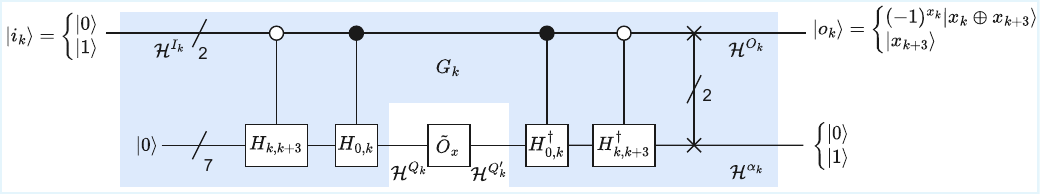}
        \caption{Circuit diagram of the quantum subroutines $G_k$. The unitary $H_{i,j}$ satisfies $H_{i,j}\ket{0} = \frac{\ket{i} + \ket{j}}{\sqrt{2}}$ and $H_{i,j}\ket{1} = \frac{\ket{i} - \ket{j}}{\sqrt{2}}$ for $0 \leq i,j \leq 6$, $i \neq j$, and can be completed arbitrarily on the other inputs. 
		  The final operation is a swap operation $\textsc{swap}: H^{O_k} \otimes H^{\alpha_k} \to H^{O_k} \otimes H^{\alpha_k}$ such that for any $a$ and $b$ in $\{0, 1\}$, $\textsc{swap} \ket{a}^{O_k} \ket{b}^{\alpha_k} = \ket{b}^{O_k} \ket{a}^{\alpha_k}$ and for any $c \in \{2, \dots, 6\}$, $\textsc{swap} \ket{a}^{O_k} \ket{c}^{\alpha_k} = \ket{a}^{O_k} \ket{c}^{\alpha_k}$.}
        \label{fig:Qsubroutine2}
\end{figure}
When acting on the query oracle $\widetilde{O}_x$, the $G_k$ act as follows:
if the qubit in $\H^{I_k}$ is $\ket{0}$, then the output qubit in $\H^{O_k}$ will be $\ket{x_k \otimes x_{k+3}}$;
however, if the qubit in $\H^{I_k}$ is $\ket{1}$, then the output in $\H^{O_k}$ is $\ket{x_{k+3}}$. 
This dependency on the value of the qubit in $\H^{I_k}$ is similar to the behaviour of the classical-deterministic process $\overline{w}_\text{Lugano}$ defined in the proof of Proposition~\ref{prop:fICO}. 
Indeed, that process was defined as a modification of the original $w_\text{Lugano}$ such that if a function $f_k$ received $0$ (resp.\ 1) in the process $w_\text{Lugano}$, it would receive $k$ (resp.\ $k+3$) in $\tilde{w}_\text{Lugano}$.

The last step is thus to compose the Lugano supermap with these quantum subroutines. 
This composition can be expressed at the level of their Choi matrices via the link product as
\begin{equation}
  \widetilde{W}_{f_{6q}} := \widetilde{W}_\text{Lugano} * G_1 * G_2 * G_3 \in \L(\H^{Q_1 Q_2 Q_3 Q'_1 Q'_2 Q'_3 F \alpha_1 \alpha_2 \alpha_3}).
\end{equation}
This process matrix characterises a quantum supermap that we write $\mathcal{S}^{f_{6q}}$, which takes three channels as input, from $\L(\H^{Q_k})$ to $\L(\H^{Q'_k})$, for $1 \leq k \leq 3$ and with a global future space $\L(\H^{F \alpha_1 \alpha_2 \alpha_3})$. 
We summarise in Table~\ref{table:QuantumLugano} its action for any $x \in \{0, 1\}^6$ on three copies of the query oracle $\widetilde{O}_x$ (with Choi matrix $\widetilde{O}_x$), which produces the output state
\begin{equation}
  \widetilde{W}_{f_{6q}} * \widetilde{\mathsf{O}}_x^{\otimes 3} \in \L(\H^{F \alpha_1 \alpha_2 \alpha_3}),
\end{equation}
where we write, with a slight abuse of notation, $\widetilde{\mathsf{O}}^{\otimes 3}_x = \widetilde{\mathsf{O}}_x \otimes \widetilde{\mathsf{O}}_x \otimes \widetilde{\mathsf{O}}_x \in \L(\H^{Q_1 Q_2 Q_3 Q'_1 Q'_2 Q'_3})$.

\begin{table}[ht]
      \caption{Values of the registers $\H^{F}$ and $\H^{\alpha_1 \alpha_2 \alpha_3}$ when the supermap characterised by $\widetilde{W}_{f_{6q}}$ acts on three copies of $\widetilde{O}_x$, for different values of $x$.}
    \label{table:QuantumLugano}
    \centering
    \begin{tabular}{c|c|c|l|c|c}
        $x_1 \oplus x_4$ & $x_2 \oplus x_5$ & $x_3 \oplus x_6$ & $\mathcal{H}^{F}$ &  $\H^{\alpha_1 \alpha_2 \alpha_3}$ & $f_{6q}(x)$ \\ \hline
        0 & 0 & 0 & $\ket{0 0 0}$ & $\ket{0 0 0}$ & 0\\
        1 & 0 & 0 & $\ket{1 x_5 0}$ & $\ket{0 1 0}$ & $x_5$ \\
        0 & 1 & 0 & $\ket{0 1 x_6}$ & $\ket{0 0 1}$ & $x_6$ \\
        0 & 0 & 1 & $\ket{x_4 0 1}$ & $\ket{1 0 0}$ & $x_4$ \\
        1 & 1 & 0 & $\ket{1 x_5 0}$ & $\ket{0 1 0}$ & $x_5$ \\
        1 & 0 & 1 & $\ket{x_4 0 1}$ & $\ket{1 0 0}$ & $x_4$ \\
        0 & 1 & 1 & $\ket{0 1 x_6}$ & $\ket{0 0 1}$ & $x_6$ \\
        1 & 1 & 1 & $\ket{1 1 1}$ & $\ket{0 0 0}$ & 1 
    \end{tabular}

\end{table}

We can see from Table~\ref{table:QuantumLugano} that by measuring the three states in $\H^{\alpha_k}$ in the computational basis $\{\ket{i}\}_{ i \in \{0, \dots, 6\}}$, one obtains the necessary information to extract the value of $f(x)$ from the register $\H^F$, for any $x \in \{0, 1\}^6$. 
Indeed, if the measurement result is $000$, then measuring any qubit of $\H^F$ in the computational basis outputs the value of $f(x)$. 
If the measurement result contains a $1$, then measuring, in the computational basis, the qubit of $\H^F$ that is in the same position as that $1$ also gives the values of $f(x)$. 
This measurement procedure can be seen as a channel $\M : \L(\H^{F \alpha_1 \alpha_2 \alpha_3}) \to \L(\H^{\tilde{F}})$ with $\dim(\H^{\tilde{F}}) = 2$, which post-processes the output such that measuring $\widetilde{W}_{f_{6q}} * \widetilde{\mathsf{O}}^{\otimes 3}_x * \mathsf{M} \in \L(\H^{\tilde{F}})$ in the computational basis outputs $f_{6q}(x)$ with probability one, concluding the proof.
\end{proof}

In this proof, the final composition with the channel $\M$ described above highlights that the causally indefinite computation we described does not compute $f_{6q}$ ``cleanly''. 
Indeed, because the state $\widetilde{W}_{f_{6q}} * \widetilde{\mathsf{O}}^{\otimes 3}_x \in \L(\H^{F \alpha_1 \alpha_2 \alpha_3})$ contains information about $f_{6q}(x)$ but also on $x$ (see Table~\ref{table:QuantumLugano}), this channel $\M$ erases this additional information about $x$. 
This extra information prevents us from directly amplifying this quantum query separation, as it prevents us from using this supermap as a coherent subroutine to compute a recursive composition of $f_{6q}$.

\end{document}